\documentclass[a4paper,11pt, headings=small, DIV=classic, numbers=enddot]{scrartcl}
\usepackage[utf8]{inputenc}
\usepackage[T1]{fontenc}
\usepackage[english]{babel}

\usepackage{scrpage2}
\usepackage{authblk}
\usepackage{etoolbox}
\usepackage{amsmath}
\usepackage{amssymb}
\usepackage{amsfonts}
\usepackage[mathscr]{eucal}
\usepackage{amsthm}
\usepackage{lmodern}
\usepackage{todonotes}
\usepackage{multirow}
\usepackage{booktabs}

\usepackage[centerdot]{mathtools}
\usepackage{tikz} 
\usepackage{multirow}
\usepackage{subcaption}
\usepackage{ifthen}
\usepackage{array}
\usepackage{pdflscape}
\usepackage{csquotes}
\usepackage[abbreviations]{foreign}
\usepackage{microtype}

\usetikzlibrary{shadows}
\usetikzlibrary{decorations.pathmorphing}

\definecolor{CeruleanRef}{RGB}{12,127,172}
\KOMAoptions{DIV=last}
\usetikzlibrary{external}
\tikzsetexternalprefix{external-figures/}
\tikzset{external/only named=true}

\ihead{}
\ohead{}
\ifoot{}
\ofoot{}
\addtokomafont{disposition}{\rmfamily\bfseries}
\addtokomafont{title}{\normalfont\sffamily}
\addtokomafont{caption}{\small}
\addtokomafont{captionlabel}{\bfseries\small\rmfamily}
\addtokomafont{descriptionlabel}{\rmfamily}
\setkomafont{pageheadfoot}{\rmfamily}
\addtokomafont{pagenumber}{\rmfamily}
\setkomafont{subparagraph}{\normalfont\textit}

\setcapindent{0pt}

\definecolor{mydarkblue}{rgb}{0.25, 0.47, 0.61}
\definecolor{mygray}{rgb}{0.89, 0.89, 0.89}
\definecolor{mylightblue}{rgb}{0.91, 0.93, 0.95}
\definecolor{mydarkgray}{rgb}{0.41, 0.33, 0.33}
\definecolor{myred}{rgb}{0.85, 0.16, 0.11}
\definecolor{mygreen}{rgb}{0.16, 0.85, 0.11}
\definecolor{myyellow}{rgb}{0.8, 0.77, 0.0}

\definecolor{myviolet}{HTML}{71277a}

\definecolor{dissbasered}{HTML}{B32B27}
\definecolor{dissbaseblue}{HTML}{2775B3}
\definecolor{dissbasegreen}{HTML}{7AAE2b}
\definecolor{dissbasegray}{HTML}{6F7267}

\colorlet{dissdarkred}{dissbasered!70!black}
\colorlet{dissdarkblue}{dissbaseblue!70!black}
\colorlet{dissdarkgreen}{dissbasegreen!70!black}
\colorlet{dissdarkgray}{dissbasegray!70!black}

\colorlet{disslightred}{dissdarkred!30!white}
\colorlet{disslightblue}{dissdarkblue!30!white}
\colorlet{disslightgreen}{dissdarkgreen!30!white}
\colorlet{disslightgray}{dissdarkgray!10!white}

\colorlet{dissblue}{dissdarkblue!60!white}
\colorlet{dissred}{dissdarkred!60!white}
\colorlet{dissgreen}{dissdarkgreen!60!white}
\colorlet{dissgray}{dissdarkgray!40!white}

\tikzset{
  myshadow/.style={opacity=.4, shadow xshift=0.03, shadow yshift=-0.03},
  mynode/.style={
  fill=mygray!40!white, 
  minimum size=5pt,
  inner sep=0mm, 
  draw=mydarkgray,
  thick,
  circle,
  drop shadow=myshadow
  }, 
  mybluenode/.style={mynode,draw=mydarkblue},
  mylabel/.style={fill=none,draw=none,font=\footnotesize, rectangle, inner sep=0.75mm},
  myedge/.style={very thick, mydarkgray},
  inactive/.style={thick,mygray!60!black, dashed},
  tA/.style={mynode,draw=myred,rectangle, thick},
  rA/.style={tA, fill=myred},
  tB/.style={mynode,draw=mydarkblue,diamond, minimum size=7pt},
  rB/.style={tB, fill=mydarkblue},
  tC/.style={mynode,draw=dissdarkgreen!90!white,star, minimum size=7pt},
  rC/.style={tC, fill=dissdarkgreen!90!white},
  tD/.style={mynode,draw=myyellow,cloud, minimum size=7pt},
  rD/.style={tD, fill=myyellow},
  tE/.style={mynode,draw=myviolet, isosceles triangle, shape border rotate=90, isosceles triangle stretches, minimum size=7pt},
  rE/.style={tE, fill=myviolet},
  removed/.style={dissbasegreen, ultra thick},
  added/.style={dissbasered, ultra thick},
  untouched/.style={mydarkgray, thick}
}

\DeclareRobustCommand{\nodetypeA}{%
\begin{tikzpicture}
  \draw node[tA,anchor=south] (v0) {};
\end{tikzpicture}}
\DeclareRobustCommand{\nodetypeB}{%
\begin{tikzpicture}
  \draw node[tB,anchor=south] (v0) {};
\end{tikzpicture}}
\DeclareRobustCommand{\nodetypeC}{%
\begin{tikzpicture}
  \draw node[tC,anchor=south] (v0) {};
\end{tikzpicture}}
\DeclareRobustCommand{\nodetypeD}{%
\begin{tikzpicture}
  \draw node[tD,anchor=south] (v0) {};
\end{tikzpicture}}
\DeclareRobustCommand{\nodetypeE}{%
\begin{tikzpicture}
  \draw node[tE,anchor=south] (v0) {};
\end{tikzpicture}}

\DeclareRobustCommand{\roottypeA}{%
\begin{tikzpicture}
  \draw node[rA,anchor=south] (v0) {};
\end{tikzpicture}}
\DeclareRobustCommand{\roottypeB}{%
\begin{tikzpicture}
  \draw node[rB,anchor=south] (v0) {};
\end{tikzpicture}}
\DeclareRobustCommand{\roottypeC}{%
\begin{tikzpicture}
  \draw node[rC,anchor=south] (v0) {};
\end{tikzpicture}}

\newcommand{\N}{\ensuremath{\mathbb{N}}}

\newcommand{\R}{\ensuremath{\mathbb{R}}}

\newcommand{\Rp}{\ensuremath{\mathbb{R}_{\geq 0}}}
\newcommand{\Zp}{\ensuremath{\mathbb{Z}_{\geq 0}}}

\DeclareMathOperator{\conv}{conv}

\DeclareMathOperator{\dist}{dist}

\DeclareMathOperator{\Proj}{Proj}

\DeclarePairedDelimiter{\abs}{\lvert}{\rvert}

\newcommand{\tforall}{\text{for all}}
\newcommand{\tand}{\text{and}}

\newcommand{\tandall}{\text{and all}}
\newcommand{\tandinteger}{\text{and integer}}

\newcommand{\srsum}[1]{\smashoperator[r]{\sum_{#1}}}

\newcommand{\slrsum}[1]{\smashoperator[lr]{\sum_{#1}}}

\newcommand{\transp}{\mathrm{T}}

\makeatletter
\def\moverlay{\mathpalette\mov@rlay}
\def\mov@rlay#1#2{\leavevmode\vtop{%
\baselineskip\z@skip \lineskiplimit-\maxdimen
\ialign{\hfil$\m@th#1##$\hfil\cr#2\crcr}}}
\newcommand{\charfusion}[3][\mathord]{
#1{\ifx#1\mathop\vphantom{#2}\fi
\mathpalette\mov@rlay{#2\cr#3}
}
\ifx#1\mathop\expandafter\displaylimits\fi}
\makeatother

\newcommand{\cset}[2]{\{#1 \mid #2\}}

\newcommand{\tuset}[1]{\{#1\}}

\newcommand{\bigsetm}[2]{\bigl\{#1\mathrel{}\bigm\lvert\mathrel{} #2\bigr\}}
\newcommand{\Bigsetm}[2]{\Bigl\{#1\mathrel{}\Bigm\lvert\mathrel{} #2\Bigr\}}

\newlength{\lpRhsSeparation}
\newboolean{lpNeedStepRefcounter}

\newcommand{\objn}[2]{#1\quad & \mathrlap{#2} &&&\hspace{\lpRhsSeparation}&}
\newcommand{\obj}[2]{\objn{#1}{#2}\ifthenelse{\boolean{lpNeedStepRefcounter}}{\refstepcounter{equation}}{}\tag{\theequation}}

\newcommand{\thmicon}{}
\newcommand{\thmfont}[1]{#1}
\theoremstyle{plain}
\newtheorem{theorem}{\thmicon\thmfont{Theorem}}[section]
\newtheorem{lemma}[theorem]{\thmicon\thmfont{Lemma}}
\newtheorem{corollary}[theorem]{\thmicon\thmfont{Corollary}}
\newtheorem{definition}[theorem]{\thmicon\thmfont{Definition}}
\newtheorem{property}[theorem]{\thmicon\thmfont{Property}}

\theoremstyle{definition}

\theoremstyle{remark}

\newcounter{property}

\newcommand{\eps}{\varepsilon}
\renewcommand{\epsilon}{\eps}

\usetikzlibrary{decorations}
\usetikzlibrary{intersections}
\usetikzlibrary{shapes}
\usetikzlibrary{patterns}
\usetikzlibrary{calc}
\usetikzlibrary{positioning}
\tikzset{
    solid part/.style={%
        postaction={solid, decorate, draw,
            decoration={moveto,
                pre=curveto, 
                post=curveto, 
                pre length=#1,
                post length=0}}
    }
}   

\tikzset{
    small circles/.style={draw, circle, minimum
        size=0.3em, inner sep=0.1em, font=\scriptsize},
    copied/.style={fill=white!90!black, small circles}
}

\tikzset{
	graphnode/.style  = {line width=0pt, circle, draw=none, fill, minimum size=5pt, inner sep=0cm},
	graphedge/.style  = {very thick, >=stealth},
	genlabel/.style   = {font={\footnotesize}},
	graphlabel/.style = {genlabel}	
}

\tikzset{
  dissnode/.style = {label distance=0.0mm, fill, draw, circle, inner sep=0pt, minimum size=8pt},
  rectdissnode/.style = {dissnode, rectangle, minimum size=7pt},
  diamdissnode/.style = {dissnode, diamond, minimum size=10pt},
  disslabel/.style = {rectangle, inner sep=0.5mm, minimum size=0mm, fill=none, draw=none, font=\footnotesize}
}

\tikzset{
  fillednode/.style = {thick, draw, font=\small, fill=lightgray, minimum size=15pt, circle, inner sep=0mm} }

\tikzset{%
  add/.style args={#1 and #2}{dashed,to path={%
      ($(\tikztostart)!-#1!(\tikztotarget)$)--($(\tikztotarget)!-#2!(\tikztostart)$)%
      \tikztonodes}}
}

\newcolumntype{m}[1]{>{\centering\let\newline\\\arraybackslash\hspace{0pt}}p{#1}}
\newcolumntype{f}[1]{>{\raggedright\let\newline\\\arraybackslash\hspace{0pt}}p{#1}}

\newcommand{\xij}{x_{ij}}
\newcommand{\fij}{f_{ij}}
\newcommand{\fji}{f_{ji}}

\newcommand{\edgeij}{\{i,j\}}

\newcommand{\rootset}{\mathfrak{R}}
\newcommand{\terms}{\mathfrak{T}}
\newcommand{\nonrootterms}{\terms\setminus\rootset}
\newcommand{\LP}{\mathfrak{L}}
\newcommand{\sfpoly}{\mathfrak{F}}
\newcommand{\steinerset}{\mathfrak{N}}
\newcommand{\relcuts}{\mathfrak{S}}
\newcommand{\nonrelcuts}{\overline{\relcuts}}

\title{MIP Formulations for the Steiner Forest Problem}
\date{September 2017}

\author[1]{Daniel R. Schmidt}
\author[2]{Bernd Zey}
\author[3]{Fran\c cois Margot}

\affil[1]{%
Universit\"at zu K\"oln\\
Institut für Informatik\\
\texttt{schmidt@informatik.uni-koeln.de}
}

\affil[2]{
TU Dortmund\\
Fakult\"at f\"ur Informatik\\
\texttt{bernd.zey@tu-dortmund.de}
}

\affil[3]{
  Carnegie Mellon University\\
  Tepper School of Business\\
}

\begin{document}
\maketitle
\begin{abstract}
The Steiner Forest problem is among the fundamental network design problems. 
Finding tight linear programmming bounds for the problem is the key for both fast Branch-and-Bound algorithms and good primal-dual approximations.
On the theoretical side, the best known bound can be obtained from an integer program by Könemann, Leonardi, Schäfer and van~Zwam~\cite{KLSvZ2008a}.
It guarantees a value that is a ($2-\eps$)-approximation of the integer optimum.
On the practical side, bounds from a mixed integer program by Magnanti and Raghavan~\cite{MR2005} are very close to the integer optimum in computational experiments, but the size of the model limits its practical usefulness.
We compare a number of known integer programming formulations for the problem and propose three new formulations.
We can show that the bounds from our two new cut-based formulations for the problem are within a factor of 2 of the integer optimum.
In our experiments, the formulations prove to be both tractable and provide better bounds than all other tractable formulations.
In particular, the factor to the integer optimum is much better than 2 in the experiments.
\end{abstract}

\section{Introduction} 

The Steiner Forest problem (SFP) is one of the fundamental network design problems.
Given an undirected graph $G=(V,E)$ with an edge weight $c_e$ for each edge $e \in E$ and terminal sets $T^1,\dots,T^K \subseteq V$, it asks for a minimum weight, cycle-free subgraph of $G$ in which the nodes in each terminal set are connected.
The aim of this work is to find stronger linear programming bounds for the Steiner Forest Problem by adapting Integer Programming formulations from the Steiner Tree problem -- the special case where $K=1$.
Our motivation comes from the fact that two provenly strong Integer Programming formulations for the Steiner Tree problem have no known analogon in the Steiner Forest case. 
We aim to bridge this gap.

There are a number of known Integer Linear Programming formulations for the Steiner Forest problem:
The equivalent undirected \emph{flow}-based and \emph{cut-set}-based formulations are known to only yield weak bounds through their linear programming relaxation.
Stronger, layered formulations draw their strength from the following observation~\cite{Martin1986,BMW1989,CR1994}: 
Instead of connecting each terminal set with an undirected tree, we can equivalently bi-direct the input graph, select a root node for each terminal set and then search for arborescences (i.e., a directed trees).
Any such arborescence induces an orientation of each edge and this can be used to strenghten the linear programming relaxations.
In the Steiner \emph{Tree} case where only one terminal set -- and thus, only one root node -- exists, this process is straight-forward. 
When multiple root nodes are present, however, the different arborescences may in general impose conflicting orientations to the edges. 
This is a major additional difficulty in solving the Steiner Forest problem. 
Magnanti and Raghavan~\cite{MR2005} show how to consolidate the conflicts and obtain a very strong linear programming relaxation.
Unfortunately, the resulting formulation is (doubly) exponentially large and to date, no efficient separation procedure is known.

The issues with conflicting orientations can be avoided altogether by using strong undirected formulations. 
Goemans~\cite{Goemans1994}, Lucena~\cite{Lucena1993}, as well as Margot, Prodon and Liebling~\cite{MPL1994} independently propose an ILP formulation for the Steiner \emph{Tree} problem that builds on Edmond's complete description of the tree polytope~\cite{Edmonds2003}.
This tree-based formulation has a straight-forward extension to the Steiner \emph{Forest} problem: 
It has, however, weakly coupled layers for each terminal set that diminish the strength of the originally strong formulation. 

We propose three new formulations for the Steiner Forest problem that strive to improve on the weaknesses of the existing formulations.
Our cut-set variants of the Magnanti-Raghavan formulation reduce the number of variables and consist of efficiently separable constraints.
While they yield a weaker lower bound than the original formulation, our experiments show that it consistently provides better lower bounds than all of the other ILP formulations. 
As a second result, we propose a new formulation that is tree-based and has a stronger coupling between the layers.
Figure~\ref{fig:lp-comparison} shows a comparison of the formulations on a small example instance.

\begin{figure}[t]
\centering
\begin{subfigure}{0.3\textwidth}
\centering
\begin{tikzpicture}
  \useasboundingbox (-0.6,-0.6) rectangle (1.6,1.6);
  \begin{scope}[scale=1]
  \draw (0,0) node[tA] (v1) {};
  \draw (0,1) node[tA] (v2) {};
  \draw (1,1) node[tA] (v3) {};
  \draw (1,0) node[tA] (v4) {};
  \end{scope}
  
  \begin{scope}[myedge]
  \draw (v1) -- (v2);
  \draw (v1) -- (v4);
  \draw (v2) -- (v3);
  \draw (v3) -- (v4);
  \end{scope}
\end{tikzpicture}
\caption{instance (A)}
\end{subfigure}
\begin{subfigure}{0.3\textwidth}
  \centering
  \begin{tikzpicture}
   \useasboundingbox (-0.6,-0.6) rectangle (1.6,1.6);
  \begin{scope}[scale=1]
    \draw (0,0) node[tB] (v1) {};
    \draw (0,1) node[tA] (v2) {};
    \draw (1,1) node[tB] (v3) {};
    \draw (1,0) node[tA] (v4) {};
  \end{scope}
  \begin{scope}[myedge]
  \draw (v1) -- (v2);
  \draw (v1) -- (v4);
  \draw (v2) -- (v3);
  \draw (v3) -- (v4);
  \end{scope}
\end{tikzpicture}
\caption{instance (B)}
\end{subfigure}
 \begin{subfigure}{0.3\textwidth}
   \centering
  \begin{tikzpicture}
  \useasboundingbox (-0.1,-0.1) rectangle (2.1,2.1);
  \begin{scope}[scale=2]
    \draw (0,0) node[tB] (v1) {};
    \draw (0,1) node[tA] (v2) {};
    \draw (1,1) node[tC] (v3) {};
    \draw (1,0) node[tD] (v4) {};
  \end{scope}
  \begin{scope}[xshift=0.5cm,yshift=0.5cm] 
    \draw (0,0) node[tC] (v5) {};
    \draw (0,1) node[tD] (v6) {};
    \draw (1,1) node[tB] (v7) {};
    \draw (1,0) node[tA] (v8) {};
  \end{scope}
  \begin{scope}[myedge]
  \draw (v1) -- (v2);
  \draw (v1) -- (v4);
  \draw (v2) -- (v3);
  \draw (v3) -- (v4);
  
  \draw (v5) -- (v6);
  \draw (v5) -- (v8);
  \draw (v6) -- (v7);
  \draw (v7) -- (v8);
  
  \draw (v1) -- (v5);
  \draw (v2) -- (v6);
  \draw (v3) -- (v7);
  \draw (v4) -- (v8);
  \end{scope}
\end{tikzpicture}
\caption{instance (C)}
\end{subfigure}\\[\baselineskip]

\begin{subfigure}{\textwidth}
\centering
\begin{tabular}{rrrr}
\toprule
formulation        & A & B & C\\ 
\cmidrule(rl){1-1}\cmidrule(rl){2-4}
basic flow/cut                                           & 2 & 2 & 4\\
\addlinespace
lifted cut~\cite{KLSvZ2008a}                             & 2 & 2 & 4\\
layered directed~\cite{Martin1986,BMW1989,CR1994}        & 3 & 2 & 4\\
layered tree-based~\cite{Goemans1994,Lucena1993,MPL1994} & 3 & 2 & 4\\
full directed flow-based~\cite{MR2005}                   & 3 & 3 & 6\\
\addlinespace
our tree-based                                           & 3 & 2.67 & 5\\ 
our extended cut-based                                   & 3 & 2.5  & 5.14\\
our strengthened extended cut-based                      & 3 & 3  & 6 \\
our extended flow-based                                  & 3 & 3  & 6 \\
\addlinespace
\emph{integer optimum} & \emph{3} & \emph{3} & \emph{7}\\
\bottomrule
\end{tabular}
\caption{Optima of the linear programming relaxations.}
\end{subfigure}
\caption{\label{fig:lp-comparison}%
A comparison of lower bounds obtained from linear relaxations. 
The terminal sets of the three Steiner Forest instances are depicted in different colors and shapes 
({\nodetypeA}, {\nodetypeB}, {\nodetypeC}, and {\nodetypeD}).
All edges have unit cost.}
\end{figure}

\subsection{Notation} 
Throughout, let $G=(V,E)$ be an undirected, simple graph with
nodes $V=\tuset{1,\dots,n}$ and edges $E=\tuset{e_1,\dots,e_m}$.
A \emph{cut-set} in $G$ is a subset $S \subseteq V$ of the nodes of $G$.
Any cutset $S\subseteq V$ induces a \emph{cut} $\delta(S) :=
\cset{\edgeij \in E}{\abs{\edgeij \cap S} = 1}$.
If $S=\{i\}$ is a singleton, we abbreviate $\delta(i) := \delta(\{i\})$.
If $D=(V,A)$ is a directed graph, we distinguish the \emph{outgoing cut}
$\delta^+(S) = \cset{(i,j) \in A}{i \in S\ \tand\ j \in V\setminus S}$ and
the \emph{incoming cut} $\delta^-(S) = \cset{(i,j) \in A}{i \in V\setminus S\ \tand\ j \in S}$.
Given a finite, ordered set $A$ and an arbitrary set $B \subseteq \R$, we interpret functions
$x: A \to B$ as vectors $x = (x_a)_{a \in A} \in B^A$ and vice-versa.
Furthermore, we write $x(A') := \sum_{a \in A'} x(a)$ for a subset $A' \subseteq A$.
Finally, for~$k \in \Zp$, let~$[k] := \tuset{1,\dots,k}$.

\subsection{The Steiner Forest Problem}
Consider the undirected graph $G=(V,E)$ and let $T^1,\dots,T^K \subseteq V$ be $K \in \N$ terminal sets.
A \emph{feasible} Steiner Forest of $(G, T^1,\dots,T^K)$ is a forest $(V_F
\subseteq V, E_F \subseteq E)$ in $G$ that, for all $k \in [K]$, contains an
$s$-$t$-path for all $s,t \in T^k$.
A feasible forest $(V_F, E_F)$ is optimum with respect to edge weights $c \in \Rp^E$ if it minimizes the total cost $\sum_{e \in E_F} c_e$.
Without loss of generality, we can assume that the terminal sets are pairwise disjoint: 
If $T^k$ and $T^\ell$ share at least one node, then any forest is feasible for
$T^1,\dots,T^K$ if and only if it is feasible for the instance where $T^k$ and $T^\ell$ are replaced by $T^k\cup T^\ell$.
We denote the set of all terminal nodes by $\terms := T^1 \cup \dots \cup T^K$ and
write $\tau(t) := k$ to denote the index of the unique terminal set that contains the terminal $t \in \terms$.
Furthermore, we say that the non-terminal nodes $\steinerset := V\setminus \terms$ are \emph{Steiner Nodes}. 
For each terminal set $T^k$, $k \in [K]$, we select an arbitrary node $r^k$ as a fixed root node and define $\rootset := \{r_1,\dots,r^k\}$ to be the set of all root nodes.
 
 A cut-set $S \subseteq V$ is relevant for the terminal set $T^k$ if it separates $r^k$ from some terminal $t \in T^k$, i.e., if $r^k \in S$ but $t \not\in S$ for some $t \in T^k$.
 We write $\relcuts^k$ for the set of all cut-sets that are relevant for $T^k$ and $\relcuts := \relcuts^1 \cup \dots \cup \relcuts^K$ for the set of all relevant cut-sets.
    
%\subsection{Linear Programming Relaxations}
%In the sequel, let us write $F^x \subseteq E$ to denote the set of edges corresponding to a characteristic vector $x \in \{0,1\}^E$. 
The convex hull of the characteristic vectors of all feasible Steiner Forests of $(G,\terms)$ is called the \emph{Steiner Forest Polytope}
\begin{align*}
\sfpoly(G, \terms) := \conv\bigsetm{x \in \{0,1\}^E}{x\text{ induces a feasible Steiner Forest of } (G,\terms)}.
\end{align*}
If $P := \cset{(x,y) \in \R^{n_1+n_2}}{Ax + By = d}$ is an arbitrary polyhedron given in its linear description, we write $\Proj_x(P) := \cset{x \in \R^{n_1}}{\exists\ y \in \R^{n_2}\ \text{s.t.}\ (x,y) \in P}$ to denote the projection of $P$ onto the $x$ variables.

\section{A Catalogue of State of the Art Steiner Forest Formulations}
In the spirit of Goeman's and Myung's~\emph{Catalogue of Steiner Tree Formulations}~\cite{GM1993} we first give an overview of Steiner \emph{Forest} MIP formulations from the literature.
An experimental comparison of the formulations in this section can be found in Section~\ref{sec:experiments}.
\subsection{Basic Formulations}
The Steiner Forest problem has a straight-forward flow formulation:
We introduce a variable $x_{ij}$ for each edge $\{i,j\} \in E$ and define two flow variables $f_{ij}^t, f^t_{ji}$ for each non-root terminal $t \in \nonrootterms$.  
Then, a selection of edges induced by the $x$ variables forms a feasible Steiner Forest, if it allows us to send one unit of flow from the $k$-th root $r^k$ to any terminal $t \in T^k$, for all $k \in [K]$. 
This leads to the formulation 
\begin{align}
\min\cset{c^\transp x}{(x,f) \in \LP^{uf}\ \tandinteger} \tag{IPuf}\label{ip:undir-flow}
\end{align}
 where 
\begin{subequations}\label{lp:undir-flow}
\begin{align}
  \LP^{uf} := \Bigl\{ (x,f) \mathrel{}\Bigm\lvert 
         f^t(\delta^+(i)) - f^t(\delta^-(i)) &=
                                                 \begin{cases}
                                                       1,&\text{if $i = r_{\tau(t)}$}\\
                                                      -1,&\text{if $i = t$}\\ 
                                                       0,&\text{otherwise}
                                                \end{cases} && \begin{aligned} &\tforall\ i \in V\\ &\tandall\ t \in \nonrootterms\end{aligned}\label{ip:undir-flow-1}\\
       \fij^t + \fji^t                &\leq \xij		&& \begin{aligned} &\tforall\ \edgeij \in E\\ &\tandall\ t\in\nonrootterms\end{aligned}\label{ip:undir-flow-2}\\[\baselineskip]
       \fij^t, \fji^t                   &\in [0,1]		&& \begin{aligned} &\tforall\ \edgeij \in E\\ &\tandall\ t\in\nonrootterms\end{aligned}\label{ip:undir-flow-3}\\
       x_{ij}                          &\in [0,1]      	&& \tforall\ \{i,j\} \in E \label{ip:undir-flow-4}
      \,\Bigr\}. 
\end{align}
\end{subequations}
The program can be reformulated with cut-conditions, using again a variable $x_{ij}$ for each edge~$\{i,j\} \in E$:
A selection of edges is feasible if and only if every relevant cut is covered by at least one edge.
The corresponding ILP formulation reads 
\begin{align}
\min\cset{c^\transp x}{x \in \LP^{uc}\ \tandinteger}\tag{IPuc} \label{ip:undir-cut}
\end{align}
 where
\begin{subequations}\label{lp:undir-cut}
\begin{align}
  \LP^{uc} := \Bigl\{ x \mathrel{}\Bigm\lvert
      x(\delta(S)) \geq 1\ \tforall\ S \in \relcuts,\ \text{where}\ x &\in [0,1]^E
      \Bigr\}\label{ip:undir-cut-1}
\end{align}
\end{subequations}
By Gale's extension~\cite{Gale1957} of the max-flow min-cut theorem, the projection
of~$\LP^{uf}$ onto the $x$-variables is exactly~$\LP^{uc}$.
Thus, the linear programming relaxations of~\eqref{ip:undir-flow} and~\eqref{ip:undir-cut} yield the same
lower bound. 
\begin{lemma}
\label{lemma:uc:vs:uf}
$\Proj_x(\LP^{uf}) = \LP^{uc}$. \qed
\end{lemma}

\subsection{Consistent Edge Orientations for Each Terminal Set}

While the linear programming bounds obtained from~\eqref{ip:undir-flow} and~\eqref{ip:undir-cut} are known to be weak~\cite{CR1994}, the formulations can be potentially strengthened.
To this aim, we impose that the edges used to connect each terminal set induce a consistent orientation of the edges in the following way. 
\begin{align}
  \min\cset{c^\transp x}{(x,f) \in \LP^{df}\ \tandinteger}\tag{IPdf} \label{ip:dir-flow}
\end{align}
where
\begin{subequations}\label{lp:dir-flow}
  \begin{align}
  \LP^{df} := \Bigl\{ (x,f) \mathrel{}\Bigm\lvert
   f^t(\delta^+(i)) - f^t(\delta^-(i)) &=
    \begin{cases}
         1,&\text{if $i = r_{\tau(t)}$}\\
        -1,&\text{if $i = t$}\\
         0,&\text{otherwise}
    \end{cases}                && \begin{aligned} &\tforall\ i \in V\\ &\tandall\ t \in \nonrootterms\end{aligned}\label{ip:dir-flow-1}\\
       0 \leq \fij^s + \fji^t &\leq \xij && \begin{aligned} &\tforall\ \edgeij \in E,\\ &\text{all}\ s,t \in T^k\setminus\{r^k\}\\ &\tandall\ k \in [K]\end{aligned}\label{ip:dir-flow-2}\\[\baselineskip]
       \fij^t, \fji^t &\in [0,1] && \begin{aligned} &\tforall\ \edgeij \in E\\ &\tandall\ t\in\nonrootterms\end{aligned}\label{ip:dir-flow-3}\\
       x_{ij} &\in [0,1] && \tforall\ \{i,j\} \in E\, \Bigr\}. \label{ip:dir-flow-4}
   \end{align}
\end{subequations}
The corresponding Steiner Tree formulation of~\eqref{ip:dir-flow} goes back to \cite{Martin1986,BMW1989}.
Here, the important change is in constraint~\eqref{ip:dir-flow-2} that forces any two flows $f^s$ and $f^t$ to agree on the orientation of any edge $\{i,j\}$ given that $s$ and $t$ belong to the same terminal set~$T^k$.
Again, this formulation can be turned into an equivalent cut-set formulation in which we have additional decision variables $y^k_{ij}, y^k_{ji}$ for each $k \in [K]$ and each $\edgeij \in E$.
\begin{align}
  \min\cset{c^\transp x}{ (x,y) \in \LP^{dc}\ \tandinteger} \tag{IPdc} \label{ip:dir-cut}
\end{align}
where
\begin{subequations}\label{lp:dir-cut}
\begin{align}
  \LP^{dc} := \Bigl\{ (x,y) \mathrel{}\Bigm\lvert y^k(\delta^+(S)) &\geq 1 && \begin{aligned} &\tforall\ k \in [K]\\ &\tandall\ S \in \relcuts^k\end{aligned}\label{ip:dir-cut-1}\\
       y^k_{ij} + y^k_{ji} &\leq x_{ij} && \begin{aligned} &\tforall\ \{i,j\} \in E\\ &\tandall\ k \in [K]\end{aligned}\\
       y^k_{ij}, y^k_{ji} &\in [0,1] && \begin{aligned} &\tforall\ \{i,j\} \in E\\ &\tandall\ k \in [K]\end{aligned}\\
       x_{ij} &\in [0,1] && \tforall\ \{i,j\} \in E\,\Bigr\}.
\end{align}
\end{subequations}
Invoking Gale's Theorem~\cite{Gale1957}, we can find a consistent orientation of the edges for each terminal set $T^k$ if and only if all cuts in~$\relcuts^k$ have at least one outgoing edge.
Thus, the projection of $\LP^{df}$ onto the $x$-variables is exactly the same as the projection of $\LP^{dc}$ onto the $x$-variables and we see that the linear programming relaxations of~\eqref{ip:dir-flow} and of~\eqref{ip:dir-cut} yield the same lower bound.
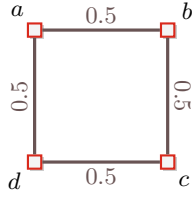
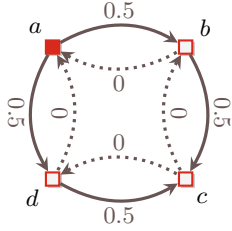
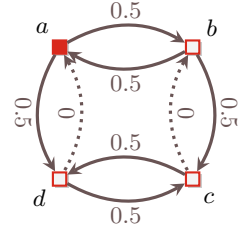
\begin{figure}[t]
\begin{subfigure}[t]{0.3\textwidth}
  \centering
  \begin{tikzpicture}
      \useasboundingbox (-0.7,-0.7) rectangle (2,2);
  \begin{scope}[scale=1.75]
    \draw (0,0) node[tA,label={[mylabel]210:$d$}] (v1) {};
    \draw (0,1) node[tA,label={[mylabel]120:$a$}] (v2) {};
    \draw (1,1) node[tA,label={[mylabel]30:$b$}] (v3) {};
    \draw (1,0) node[tA,label={[mylabel]300:$c$}] (v4) {};
  \end{scope}
  \begin{scope}[myedge, mylabel, sloped, above]
  \draw (v1) -- node{0.5} (v2);
  \draw (v1) -- node[below]{0.5} (v4);
  \draw (v2) -- node{0.5} (v3);
  \draw (v3) -- node{0.5} (v4);
  \end{scope}
\end{tikzpicture}
\caption{A feasible solution for~\eqref{lp:undir-cut}.
The edges $\{c,d\}$ and $\{b,c\}$ cover the cuts $\delta(\{a,b,c\})$ and $\delta(\{a,b,d\})$.}
\end{subfigure}\qquad
\begin{subfigure}[t]{0.3\textwidth}
  \centering
  \begin{tikzpicture}
   \useasboundingbox (-0.7,-0.7) rectangle (2,2);
  \begin{scope}[scale=1.75]
    \draw (0,0) node[tA,label={[mylabel]210:$d$}] (v1) {};
    \draw (0,1) node[rA,label={[mylabel]120:$a$}] (v2) {};
    \draw (1,1) node[tA,label={[mylabel]30:$b$}] (v3) {};
    \draw (1,0) node[tA,label={[mylabel]300:$c$}] (v4) {};
  \end{scope}
  \begin{scope}[myedge, bend left, ->, >=stealth, sloped, above]
  \draw[bend right] (v2) to node[mylabel,below]{0.5} (v1);
  \draw[bend right] (v1) to node[mylabel,below]{0.5} (v4);
  \draw             (v2) to node[mylabel]{0.5} (v3);
  \draw             (v3) to node[mylabel]{0.5} (v4);
  \end{scope}
  \begin{scope}[myedge, bend right, <-, >=stealth, dotted, sloped, below]
  \draw[bend left] (v2) to node[mylabel]{0} (v1);
  \draw            (v2) to node[mylabel]{0} (v3);
  \draw            (v3) to node[mylabel,above]{0} (v4);
  \draw[bend left] (v1) to node[mylabel,above]{0} (v4);
  \end{scope}
\end{tikzpicture}
  \caption{An \emph{infeasible} solution for~\eqref{lp:dir-cut}.
   The \emph{arcs} $(d,c)$ and $(b,c)$ cover the cut~$\delta(\{a,b,d\})$,
   but not the cut~$\delta(\{a,b,c\})$. 
  }
\end{subfigure}\qquad
\begin{subfigure}[t]{0.3\textwidth}
  \centering
  \begin{tikzpicture}
      \useasboundingbox (-0.7,-0.7) rectangle (2,2);
  \begin{scope}[scale=1.75]
    \draw (0,0) node[tA,label={[mylabel]210:$d$}] (v1) {};
    \draw (0,1) node[rA,label={[mylabel]120:$a$}] (v2) {};
    \draw (1,1) node[tA,label={[mylabel]30:$b$}] (v3) {};
    \draw (1,0) node[tA,label={[mylabel]300:$c$}] (v4) {};
  \end{scope}
  \begin{scope}[myedge, bend left, ->, >=stealth, sloped, above]
  \draw[bend right] (v2) to node[mylabel,below]{0.5} (v1);
  \draw[bend right] (v1) to node[mylabel,below]{0.5} (v4);
  \draw             (v2) to node[mylabel]{0.5} (v3);
  \draw             (v3) to node[mylabel]{0.5} (v4);
  \end{scope}
  \begin{scope}[myedge, bend right, <-, >=stealth, dotted, sloped, below]
  \draw[bend left] (v2) to node[mylabel]{0} (v1);
  \draw[solid]     (v2) to node[mylabel]{0.5} (v3);
  \draw            (v3) to node[mylabel,above]{0} (v4);
  \draw[bend left,solid] (v1) to node[mylabel,above]{0.5} (v4);
  \end{scope}
\end{tikzpicture}
\caption{A feasible solution for~\eqref{lp:dir-cut}. 
Additional capacity is needed on $(c,d)$ and $(b,a)$ to cover all relevant cuts.
}
\end{subfigure}
\caption{\label{fig:cut-undir-vs-dir}An example where the directed cut-set relaxation~\eqref{lp:dir-cut} yields a stronger linear programming bound than the undirected cut-set relaxation~\eqref{lp:undir-cut}.
The instance has a single terminal set that contains all four nodes of the graph. 
Node $a$ has been chosen as the root.
We assume unit costs on the edges.
}
\end{figure} 
If all $|T^k|=2$ for all $k \in [K]$, then~$\LP^{df}$ and~$\LP^{dc}$ are equivalent to~$\LP^{uf}$ and~$\LP^{uc}$, respectively. 
Otherwise, Figure~\ref{fig:cut-undir-vs-dir} shows why the linear programming bound obtained from the former two formulations is potentially stronger than the one obtained from the latter.
On the other hand, the bounds obtained from $\LP^{df}$ and~$\LP^{dc}$ cannot be weaker than the ones from $\LP^{uf}$ and~$\LP^{uc}$, respectively: 
The constraints~\eqref{ip:dir-flow-1}--\eqref{ip:dir-flow-4} are a subset of the constraints~\eqref{ip:undir-flow-1}--\eqref{ip:undir-flow-4} and thus $\LP^{df}$ is contained in $\LP^{uf}$.

\begin{lemma}
\label{lemma:df:vs:dc}
$\Proj_x(\LP^{df}) = \text{Proj}_x(\LP^{dc})$.\qed
\end{lemma}
\begin{lemma}
\label{lemma:uf:vs:df}
$\Proj_x(\LP^{uf}) \supsetneq \text{Proj}_x(\LP^{df})$.\qed
\end{lemma}

\subsection{Edmond's Tree Polytope}

Another Steiner Tree ILP formulation~\cite{Goemans1994,Lucena1993,MPL1994} is based on Edmond's complete description of the tree polytope~\cite{Edmonds2003}. 
It is a natural idea to generalize this formulation to forests by simply enforcing that any connected component of the solution is a Steiner tree.
To this end, we introduce a variable $x^k_{ij}$ for each $k\in [K]$ and each $\edgeij \in E$ that select the edges used to connect $T^k$.
A new variable $y^k_i$ exists for each node $i \in V$ and each $k \in [K]$ and decides whether node $i$ is connected to the Steiner Tree of terminal set $T^k$.
The coupling variables $x_{ij}$ for each edge $\edgeij \in E$ decide if the edge~$\edgeij$ is included in the forest.
\begin{align}
  \min\cset{c^\transp x}{(x,y) \in \LP^{lt}\ \tandinteger} \tag{IPlt} \label{ip:layered-tree}
\end{align}
where 
\begin{subequations}
\begin{align}
\LP^{lt} := \Bigl\{ (x,y) \mathrel{}\Bigm\lvert
     y^k(V) - x^k(E)    &= 1           && \tforall\ k \in K\label{ip:layered-tree-1}\\
     y^k(S) - x^k(E[S]) &\geq y^k_i      && \begin{aligned}
                                              &\tforall\ k \in K\\ 
                                              &\tandall\ S \subseteq V\\ 
                                              &\tandall\ i \in S
                                            \end{aligned}\label{ip:layered-tree-2}\\
     x^k_{ij}           &\leq x_{ij}   && \begin{aligned}
                                              &\tforall\ \{i,j\} \in E\\ 
                                              &\tandall\ k \in [K]
                                             \end{aligned}\label{ip:layered-tree-3}\\
     y^k_i              &= 1           && \begin{aligned}
                                              &\tforall\ k \in [K]\\ 
                                              &\tandall\ i \in V
                                             \end{aligned}\label{ip:layered-tree-4}\\
     x^k_{ij}           &\in \{0,1\}   && \begin{aligned}
                                              &\tforall\ k \in [K]\\ 
                                              &\tandall\ \{i,j\} \in E\\
                                             \end{aligned}\label{ip:layered-tree-5}\\
     y^k_i              &\in \{0,1\}   && \begin{aligned}
                                              &\tforall\ k \in [K]\\ 
                                              &\tandall\ i \in V\\
                                            \end{aligned}\label{ip:layered-tree-6}\\
     x_{ij}             &\in \{0,1\}   && \tforall\ \{i,j\} \in E\,\Bigr\}. \label{ip:layered-tree-7}
\end{align}        
\end{subequations}                     
Here, constraint~\eqref{ip:layered-tree-2} ensures that the choice of the $y^k$ is consistent and that the graph that connects the terminals in $T^k$ is cycle-free.
Together with constraint~\eqref{ip:layered-tree-1}, this means that the graph connecting $T^k$ is a tree.

\subsection{Consistent Edge Orientations Across all Terminal Sets}

The strengthened formulations in the previous sections rely on a layering argument: 
They impose that each terminal set -- each layer -- is connected by a Steiner Tree and their
additional strength comes solely from an independent strengthening of each
layer.
As can be seen in Figure~\ref{fig:consistent-orientations}, the coupling between the layers is weak, however. 
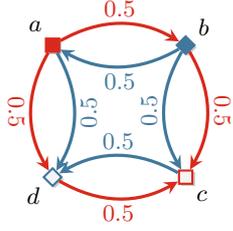
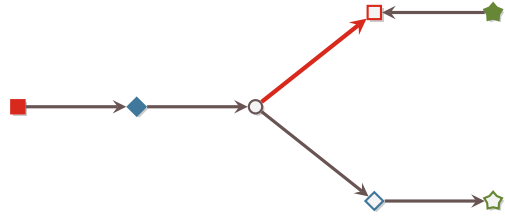
\begin{figure}
\begin{subfigure}[t]{0.45\textwidth}
  \centering
  \begin{tikzpicture}
   \useasboundingbox (-0.7,-0.7) rectangle (2,2);
  \begin{scope}[scale=1.75]
    \draw (0,0) node[tB,label={[mylabel]210:$d$}] (v1) {};
    \draw (0,1) node[rA,label={[mylabel]120:$a$}] (v2) {};
    \draw (1,1) node[rB,label={[mylabel]30:$b$}]  (v3) {};
    \draw (1,0) node[tA,label={[mylabel]300:$c$}] (v4) {};
  \end{scope}
  \begin{scope}[myedge, bend left, ->, >=stealth, sloped, above, myred]
  \draw[bend right] (v2) to node[mylabel,below]{0.5} (v1);
  \draw[bend right] (v1) to node[mylabel,below]{0.5} (v4);
  \draw             (v2) to node[mylabel]{0.5}       (v3);
  \draw             (v3) to node[mylabel]{0.5}       (v4);
  \end{scope}
  \begin{scope}[myedge, bend right, <-, >=stealth, sloped, below, mydarkblue]
  \draw                                     (v1) to node[mylabel]{0.5}         (v2);
  \draw                                     (v2) to node[mylabel]{0.5}         (v3);
  \draw[bend left]                          (v4) to node[mylabel,above]{0.5}   (v3);
  \draw[bend left]                          (v1) to node[mylabel,above]{0.5}   (v4);
  \end{scope}
\end{tikzpicture}
  \caption{%
  Detailed picture of instance (B) from Figure~\ref{fig:lp-comparison}.
  The \textcolor{myred}{red} and \textcolor{mydarkblue}{blue} arcs form a solution for formulation~\eqref{ip:dir-flow} for the \textcolor{myred}{red} (\nodetypeA) and \textcolor{mydarkblue}{blue} (\nodetypeB) terminal set.
  Looking for a Steiner arborescence for each terminal set does not cut off a fractional optimum of cost~$2$. 
  An integer optimum has cost~$3$.}
\end{subfigure}\qquad
\begin{subfigure}[t]{0.45\textwidth}
\begin{tikzpicture}
\begin{scope}[scale=1.25,xscale=1.25]
  \draw (0, 0) node[rA]     (v1) {};
  \draw (1, 0) node[rB]     (v2) {};
  \draw (2, 0) node[mynode] (v3) {};
  \draw (3, 1) node[tA]     (v4) {};
  \draw (3,-1) node[tB]     (v5) {};
  \draw (4, 1) node[rC]     (v6) {};
  \draw (4,-1) node[tC]     (v7) {};
\end{scope}
\begin{scope}[myedge, ->, >=stealth]
  \draw (v1) to (v2);
  \draw (v2) to (v3);
  \draw[ultra thick, myred] (v3) to (v4);
  \draw (v3) to (v5);
  \draw (v6) to (v4);
  \draw (v5) to (v7);
\end{scope}
\end{tikzpicture}
\caption{An example graph with three terminal sets. 
No choice of root nodes admits an orientation of the edges such that there is a directed path from~\roottypeA\ to~\nodetypeA, from~\roottypeB\ to~\nodetypeB, and from~\roottypeC\ to~\nodetypeC.
}
\end{subfigure}
\caption{\label{fig:consistent-orientations}
Subfigure (a) shows an example where the constraint $f^s_{ij} + f^t_{ji} \leq 1$ for $s \in T^k, t \in T^{\ell}, k \not= \ell$ would cut off fractional points.
Subfigure (b) shows why this additional constraint is not valid for the Steiner Forest polytope.}
\end{figure}
One way to improve the coupling of the layers would be to enforce a consistent orientation of the edges across the terminal sets (as opposed to having a consistent orientation of each terminal set only).
At first glance, we could try to achieve such an orientation by including constraint~\eqref{ip:dir-flow-2} in~\eqref{ip:dir-flow} for terminals that do not lie in the same terminal set. 
Unfortunately, this extension cuts of optimum integer solutions, see Figure~\ref{fig:consistent-orientations}: 
There are instances where the edges cannot be oriented such that they point away from all root nodes. 
In other words: Generally, there does not exist an orientation of the edges
of a feasible Steiner Forest such that there is a directed path from $r^k$ to
all~$t \in \nonrootterms$ \emph{for all~$k \in [K]$}.
This remains true even if we do not fix the choice of root nodes.

\begin{figure}
\centering
  \begin{tikzpicture}
    \begin{scope}[scale=1.75]
      \draw (0,0) node[tB,label={[mylabel]210:$d$}] (v1) {};
      \draw (0,1) node[rA,label={[mylabel]120:$a$}] (v2) {};
      \draw (1,1) node[rB,label={[mylabel]30:$b$}]  (v3) {};
      \draw (1,0) node[tA,label={[mylabel]300:$c$}] (v4) {};
    \end{scope}
    
    \begin{scope}[yscale=0.9,xscale=1.25, xshift=3cm, yshift=1cm]
      \draw (0, 0) node[rC,label={[mylabel]210:$e$}]     (w1) {};
      \draw (1, 0) node[rD,label={[mylabel]210:$f$}]     (w2) {};
      \draw (2, 0) node[mynode,label={[mylabel]210:$g$}] (w3) {};
      \draw (3, 1) node[tC,label={[mylabel]120:$h$}]     (w4) {};
      \draw (3,-1) node[tD,label={[mylabel]210:$i$}]     (w5) {};
      \draw (4, 1) node[rE,label={[mylabel]120:$j$}]     (w6) {}; 
      \draw (4,-1) node[tE,label={[mylabel]210:$k$}]     (w7) {};
    \end{scope}
    \begin{scope}[myedge, ->, >=stealth]
      \draw (w1) to (w2);
      \draw (w2) to (w3);
      \draw (w3) to (w4);
      \draw (w3) to (w5);
      \draw (w4) to (w6);
      \draw (w5) to (w7);
    \end{scope}
    
    \begin{scope}[myedge, ->, >=stealth]
      \draw (v2) -- (v1);
      \draw (v2) -- (v3);
      \draw (v3) -- (v4); 
    \end{scope}
  \end{tikzpicture}
\caption{\label{fig:mr-orientation}
An instance of the Steiner Forest problem that has five terminal sets (\nodetypeA, \nodetypeB, \nodetypeC, \nodetypeD, and \nodetypeE).
We assume that the terminal sets are indexed in this order.
The figure shows a consistent orientation of the edges in the sense of~\cite{MR2005}: 
All arcs in the connected component on the left point away from the lowest index root~$a$~(\roottypeA); all arcs in the connected component on the right point away from the lowest index root~$e$~(\roottypeC).}
\end{figure}
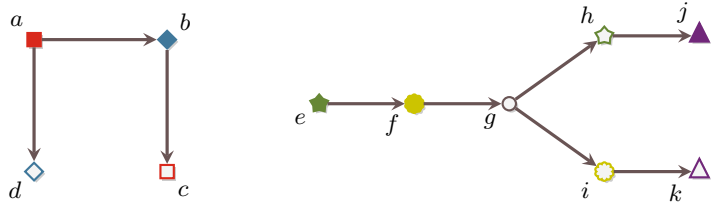
On the other hand, Magnanti and Raghavan~\cite{MR2005} observe that indeed any Steiner Forest can be oriented such that across all terminal sets, all edges consistently point away from \emph{some} root node:
Since each connected component $W$ of a feasible Steiner Forest is free of cycles, we can assign an orientation to all edges such that there is a directed path from the lowest index root node $r^\ell$ in $W$ to all terminals (including other root nodes) in $W$. 
See Figure~\ref{fig:mr-orientation} for an example. 
Thus, finding a consistent orientation for a given feasible Steiner Forest is easy. 
Enforcing such an orientation in a ILP is not, however, since we do not know the connected components of an optimum Steiner Forest a priori.
%We denote the corresponding flows by $f^{r^k, t}$ and $f^{r^\ell,r^k}$, respectively.
%We now look for directed solutions $F^+$ that satisfy the following properties:
The model by Magnanti and Raghavan approaches the problem by looking for a directed graph $F=(V,A)$ with the following properties.
\begin{enumerate}
  \item For each $k \in [K]$ and each terminal $t \in T^k\setminus\{r^k\}$, there is a unique $\ell \leq k$ for which $F$ contains a directed $r^\ell$-$t$-path.
  \item If in the above condition we have $\ell < k$, then there is a directed $r^\ell$-$r^k$-path in $F$.
  \item For all edges $\{i,j\} \in E$, the digraph $F$ contains at most one of $(i,j)$ and $(j,i)$.
\end{enumerate}
If $F$ satisfies properties (1) and (2), then $F$ induces a feasible Steiner Forest. 
% In order to state the next ILP formulation, we pair the terminals in the following way.
% \begin{itemize}
%   \item Define a pair $(r^\ell, t)$ for each terminal $t \in T^k\setminus\{r^k\}$ and each root $r^\ell$ with $\ell \leq k$, and
%   \item define a pair $(r^\ell, r^k)$ for each root $r^k$ and each index $\ell < k$.
% \end{itemize}
To state the model, let now for each $k \in [K]$
\begin{align*}
  \mathcal{O}(r^k) := \bigcup_{\ell=k}^K \bigsetm{(r^\ell,t)}{t \in T^k}
\end{align*}
and define $\mathcal{D} := \mathcal{O}(r_1) \cup \dots \cup\mathcal{O}(r^k)$ be the union of the $O(r^k)$. 
Magnanti and Raghavan consider the set
\begin{align*}
\mathcal{C}:= \mathcal{O}(r_1) \times \dots \times \mathcal{O}(r^K).
\end{align*}
% For ease of notation, we write each element $((r_1,t_t),\dots,(r^K,t^K)) \in \mathcal{C}$ as a $K$-dimensional vector $(t_1,\dots,t^K)$.
%Any element of $\mathcal{C}$ defines a set of demand pairs that use all edges in the same orientation.

This leads to the formulation proposed in~\cite[Formulation (14)]{MR2005}.
The formulation uses flow variables to check the existence of the paths from conditions (1) and (2):
The flow defined by $f^{s,t}$ establishes a directed path from $s$ to $t$, where $(s,t) \in \mathcal{D}$.
Here, the variables $y_{ij}$ and $y_{ji}$ decide if the arc $(i,j)$ and the arc $(j,i)$, respectively, can carry flow, for all $\{i,j\} \in E$, while $x_{ij}$ decides if the edge $\{i,j\}$ is included in the forest.
In the form published in~\cite{MR2005}, the formulation has $O(|E|\cdot K + |\mathcal{C}|^2 + |\mathcal{C}|\cdot |E|)$ constraints.
Raghavan~\cite[Formulation (7.4)]{Raghavan1995} shows how to reduce this number to~$O(|E|\cdot K + |\mathcal{C}|\cdot |E|)$.
We use the reduced variant~\eqref{ip:mr} in our experiments.
\begin{align}
  \min\cset{c^\transp x}{ (x, f) \in \LP^{mr}\ \tandinteger} \tag{IPmr} \label{ip:mr}
\end{align}
where
\begin{subequations}
\begin{align}
  \LP^{mr} := \Bigl\{ (x,f) \mathrel{}\Bigm\lvert 
   f^{s,t}(\delta^+(i)) - f^{s,t}(\delta^-(i)) &\begin{cases}
                                               \leq 1  &\text{if $i=s$}\\
                                               \geq -1 &\text{if $i=t$}\\
                                               =0      &\text{otherwise}
                                              \end{cases} && 
                                              \begin{aligned}
                                                &\tforall\ i \in V\\
                                                &\tandall\ (s,t) \in \mathcal{D}
                                              \end{aligned}\label{ip:mr-1}\\
   -\sum_{\ell=0}^{\tau(t)}\ \srsum{\{i,t\} \in E} f^{r^\ell,t}_{it} &= -1 && 
                                              \begin{aligned}
                                                &\tforall\ t \in \nonrootterms                                                
                                              \end{aligned}\label{ip:mr-2}\\
   \srsum{\{i,t\}\in E} f^{r^\ell,t}_{it} &\leq \srsum{\{i,r^k\}\in E} f_{ir^k}^{r^\ell,r^k} &&
                                              \begin{aligned}
                                                &\tforall\ t \in \nonrootterms\\
                                                &\tandall\ \ell \leq \tau(t)
                                              \end{aligned}\label{ip:mr-3}\\[\baselineskip]
   \srsum{(s,t) \in C} f^{s,t}_{ij}  &\leq y_{ij} &&
                                               \begin{aligned}
                                                 &\tforall\ \{i,j\} \in E\\
                                                 &\tandall\ C \in \mathcal{C}
                                               \end{aligned}\label{ip:mr-4}\\
   \srsum{(s,t) \in C} f^{s,t}_{ji}  &\leq y_{ji} &&
                                               \begin{aligned}
                                                 &\tforall\ \{i,j\} \in E\\
                                                 &\tandall\ C \in \mathcal{C}
                                               \end{aligned}\label{ip:mr-5}\\
   y_{ij} + y_{ji}                   &\leq x_{ij} &&
                                                \tforall\ \{i,j\} \in E\label{ip:mr-6}\\[\baselineskip]
   \sum_{i \in V} \srsum{(s,t)\in C} f^{s,t}_{ij} &\leq 1 &&
                                                \begin{aligned}
                                                 &\tforall\ j \in V\\
                                                 &\tandall\ C \in \mathcal{C}
                                                \end{aligned}\label{ip:mr-7}\\[\baselineskip]
   f_{ij}^{s,t}, f_{ji}^{s.t} &\in [0,1]        &&
                                                \begin{aligned}
                                                &\tforall\ \{i,j\} \in E\\
                                                &\tandall\ (s,t) \in \mathcal{D}
                                                \end{aligned}\label{ip:mr-8}\\
   y_{ij}, y_{ji}, x_{ij} &\in [0,1]            && \tforall\ \{i,j\} \in E\Bigr\}
\end{align}
\end{subequations}
The constraints~\eqref{ip:mr-2} ensures that each terminal receives at least one unit of flow, i.e. they ensure property (1).
Property (2) is enforced by the constraints~\eqref{ip:mr-3}: 
They ensure that if for some $\ell < k \in [K]$, flow is sent from $r^\ell$ to $t \in T^k$, then at least the same amount of flow is sent from $r^\ell$ to $r^k$.
The constraints~\eqref{ip:mr-4}--\eqref{ip:mr-6} establish property (3). 
Magnanti and Raghavan show that the improved formulation~\eqref{ip:mr} is stronger than the undirected cut formulation~\eqref{ip:undir-cut}.
\begin{lemma}[\cite{MR2005}]\label{lemma:uc:vs:mr}
  $\LP^{uc} \supsetneq \Proj_x(\LP^{mr})$.\qed
\end{lemma}

\subsection{Other Lower Bounds}

Könemann, Leonardi, Schäfer, and van Zwam~\cite{KLSvZ2008a} propose another method to compute a lower bound for an optimum Steiner Forest.
They formulate a linear programming relaxation of an integer linear program that potentially computes suboptimal solutions. 
The linear programming relaxation, however, yields a proper lower bound on the cost of an optimum Steiner Forest.

Observe that we can equivalently state our Steiner Forest instance $(G, T^1,\dots,T^K,c)$ in the following way:
For each $k \in K$ and each $t \in T^k\setminus \{r^k\}$ define a terminal pair $(r^k, t)$.
Denote the resulting set of all terminal pairs as $R$.
Then, a Steiner Forest is feasible for $(G, T^1,\dots,T^K,c)$ if and only if it contains an undirected $s$-$t$-path for all $(s,t) \in R$.
We impose an arbitrary ordering $(s_1, t_1), \dots, (s_L, t_L), L = |R|$, on the terminal pairs.
Denote by $\dist(s,t)$ the shortest path distance from $s$ to $t$ in $G$ with respect to the edge weights~$c$.

We recall that $\relcuts$ is the set of cut-sets $S\subseteq V$ with $s \in S$ and $t \not\in S$ for some $(s,t) \in R$.
Additionally, we define $\nonrelcuts$ to be the set of those cut-sets $S \not\in \relcuts$ with $s,t \in S$ for some $(s,t) \in R$.
For all $\ell \in [L]$, the set $\mathfrak{H}^\ell \subseteq \relcuts$ denots the set of all cuts $S \in \relcuts$ that contain $s^\ell$ or $t^\ell$ and where $\ell$ is the highest index of any terminal contained in the cut. 
Likewise, $\overline{\mathfrak{H}}^\ell \subseteq \nonrelcuts$ is the set of all cuts $S \in \nonrelcuts$ that contain $s^\ell$ and $t^\ell$ where $\ell$ is the highest index of any terminal contained in the cut.

We introduce two variables $y_\ell$ and $\overline{y}_\ell$ for each terminal pair $\ell \in [L]$ and a variable $x_{ij}$ for each edge $\{i,j\} \in E$.  
\begin{align}
  \min\Bigsetm{\sum\nolimits_{\{i,j\}\in E} x_{ij} + \sum\nolimits_{\ell=1}^L \dist(s_\ell,t_\ell)(y_{\ell} + \overline{y}_{\ell})}{(x,y) \in \LP^{klvz}\ \tandinteger} \tag{IPklsvz} \label{ip:klsvz}
\end{align}
where
\begin{subequations}
\begin{align}
  \LP^{klsvz} := \Bigl\{ (x,y) \mathrel{}\Bigm\lvert 
       x(\delta(S)) + y_\ell                &\geq 1 && \begin{aligned} &\tforall\ \ell \in [L]\\ &\tandall\ S \in\mathfrak{H}^\ell\end{aligned}\label{ip:klsvz-1}\\
       x(\delta(S)) + y_\ell + \overline{y}_\ell &\geq 1 && \begin{aligned} &\tforall\ \ell \in [L]\\ &\tandall\ S \in\overline{\mathfrak{H}}^\ell\end{aligned}\label{ip:klsvz-2}\\
       x_{ij} &\in [0,1] && \tforall\ \{i,j\} \in E\\
       y_\ell, \overline{y}_\ell &\in [0,1] &&\tforall\ \ell \in [L]\,\Bigr\}.
\end{align}
\end{subequations}
The idea behind constraint~\eqref{ip:klsvz-1} is that we can either cover any relevant cut with an edge, or we opt to pay for a direct connection between the terminal pair $(s_\ell, t_\ell)$ that is responsible for the cut by setting $y_\ell=1$. 
Könemann et al. show that the optimum value of the linear programming relaxation of~\eqref{ip:klsvz} is a lower bound for the cost of an optimum Steiner Forest of the underlying instance. 
At the same time, they prove that the projection of $\LP^{klsvz}$ onto the $x$-variables is contained in~$\LP^{uc}$, i.e., the linear programming bound obtained from~\eqref{ip:klsvz} is never worse than the one obtained from~\eqref{ip:undir-cut} and~\eqref{ip:undir-flow}.
Moreover, there are instances where $\LP^{klsvz}$ yields a strictly better linear programming bound than~\eqref{ip:undir-cut} and~\eqref{ip:undir-flow}. 
\begin{lemma}[\cite{KLSvZ2008a}]\label{lemma:uc:vs:klsvz}
  The linear programming bound obtained from~\eqref{ip:klsvz} is stronger than the one obtained from~\eqref{ip:undir-cut}.\qed
\end{lemma}

\section{New Formulations for the Steiner Forest Problem}

\subsection{A Tree-Based Formulation}
\enlargethispage{2\baselineskip}
We derive another integer linear programming formulation for the Steiner Forest problem based on Edmond's tree polytope~\cite{Edmonds2003}.
This formulation does not have a layer for each terminal set, but instead strives to directly find a feasible forest.
In the formulation, we have a decision variable $x_{ij}$ for each edge $\{i,j\}$ of $G$ and a decision variable $y_i$ for each node $i$ of $G$.
Additionally, the variable $R$ encodes the number of connected components of the solution.
For each pair of terminal sets $T^k$ and $T^{\ell}$, $k, \ell \in [K]$, we introduce a variable $w_{k\ell}$ with the intuition that $w_{k\ell} = 1$ if and only if $T^k$ and $T^\ell$ lie in the same connected component of the solution.
Lastly, we need a variable $a_k$ for each $k\in [K]$ and a variable $z_{ik}$ for each pair of a node $i\in V$ and a terminal set $T^k$.
We will make sure that $z_{ik} = 1$ if node $i$ lies in the same connected component of the solution as $T^k$.
\begin{align}
\min\cset{c^\transp x}{(x,y,w,z,a,R) \in \LP^{et}\ \tandinteger} \tag{IPet}\label{ip:ext-tree}
\end{align}
where
\begin{subequations}
\begin{align}
    \LP^{et} &&:= \Bigl\{ (x,y,w,z,a,R) \mathrel{}\Bigm\lvert
    \label{ip:ext-tree-1}  y(V) - x(E)    &= R                              &&\\
    \label{ip:ext-tree-2} && y(S) - x(E[S]) &\geq y_i                         &&\tforall\ S \subseteq V,\ \tforall\ i \in S\\[1em]
    \label{ip:ext-tree-3} &&         a([K]) &= R                              &&\\
    \label{ip:ext-tree-4} &&            a_k &\leq 1 - w_{k\ell}               &&\tforall\ \ell,k \in [K] \text{ with } k < \ell\\[1em]
    \label{ip:ext-tree-5} && x_{ij} - (1-z_{ik}) - (1-z_{jl}) &\leq w_{kl}    &&\multirow{2}{5cm}[0.25cm]{$\tforall\ \{i,j\}\in E$, $\tforall\ k,\ell \in [K] \text{ with } k < \ell$}\\
    \label{ip:ext-tree-6} && x_{ij} - (1-z_{i\ell}) - (1-z_{jk}) &\leq w_{kl} &&\\[1em]
    \label{ip:ext-tree-7} && x_{ij} + z_{jk} -1 &\geq z_{ik}                  &&\multirow{2}{3.5cm}[0.25cm]{$\tforall\ \{i,j\} \in E$, $\tforall\ k \in [K]$}\\
    \label{ip:ext-tree-8} && x_{ij} + z_{ik} -1 &\geq z_{jk}                  && \\[1em]
    \label{ip:ext-tree-9} && w_{k\ell} + w_{\ell m} -1 &\leq w_{km}           &&\multirow{3}{3.5cm}{$\tforall\ k,l,m \in [K]$, $k<l<m$}\\
    \label{ip:ext-tree-10}&& w_{km} + w_{\ell m} -1 &\leq w_{k\ell}           && \\
    \label{ip:ext-tree-11}&& w_{k\ell} + w_{km} -1 &\leq w_{\ell m}           && \\[1em]
    \label{ip:ext-tree-12}&& z_{ik}             &= 1                          && \multirow{2}{4cm}[0.25cm]{$\tforall\ k\in K$, $\tforall\ i \in T^k$}\\
    \label{ip:ext-tree-13}&& y_{i}              &= 1                          && \\[1em]
                          && x_{ij}             &\in [0,1]                  && \tforall\ \{i,j\} \in E\\
                          && y_i                &\in [0,1]                  && \tforall\ i \in V\\
                          && w_{k\ell}          &\in [0,1]                  && \tforall\ k,\ell \in [K], k < \ell\\
                          && z_{ik}             &\in [0,1]                  && \tforall\ i\in V, k \in [K]\\
                          && a_k                &\in [0,1]                  && \tforall\ k \in [K]\\
                          && R                  &\in [K]                      && \Bigr\}. 
\end{align}
\end{subequations}
\newcommand{\fxystar}{F_{x^\ast,y^\ast}}
We show that the formulation~\eqref{ip:ext-tree} is indeed correct. 
In the following, let $\fxystar = (V_F, E_F)$ with $V_F:=\cset{i\in V}{y^\ast_i=1}$ and $E_F:=\cset{\{i,j\}\in E}{x^\ast_{ij}=1}$ the forest induced by a feasible integer solution $(x^\ast,y^\ast,r^\ast,w^\ast,z^\ast,R^\ast)$ to~\eqref{ip:ext-tree}.
\begin{lemma}
  Let $(x^\ast,y^\ast,a^\ast,w^\ast,z^\ast,R^\ast)$ be a feasible integer solution to~\eqref{ip:ext-tree}.
  Then $\fxystar$ consists of exactly $R^\ast$ connected components.
\end{lemma}
\begin{proof}
Since $(x^\ast, y^\ast)$ in particular satisfies the constraints~\eqref{ip:ext-tree-2} for all $S\subseteq V$ and all $i \in S$, the subgraph
$F:=\cset{e\in E}{x^\ast_e=1}$ induced by $x^\ast$ cannot contain any cycles.
Thus, we have $|V[F]| - |F| = R^\ast$ by constraint~\eqref{ip:ext-tree-1} and thus $F$ has exactly $R^\ast$ connected components.
\end{proof}
\begin{lemma}\label{lem:tree-poly-corr-path}
  Let $(x^\ast,y^\ast,a^\ast,w^\ast,z^\ast,R^\ast)$ be a feasible integer solution to~\eqref{ip:ext-tree} and let $P$ be a path in $\fxystar$.
  Assume that $z_{ik} = 1$ for some $i \in V[P]$. Then, we have $z_{jk} = 1$ for all $j \in V(P)$.
\end{lemma}
\begin{proof}
  Suppose that $P$ visits the nodes $i_1,\dots,i_s$ in that order.
  We can assume w.l.o.g. that $i=i_1$ since otherwise, we can split $P$ at $i$ and prove the claim separately for the two resulting subpaths.
  If $s=1$, then $z_{i_1,k} = z_{ik} = 1$ and the claim is true.
  Otherwise, assume that $s > 1$:
  By induction, we have $z_{i_{s-1},k}=1$ and since the edge $\{s-1,s\}$ lies on $P\subseteq\cset{e\in E}{x^\ast_e=1}$ it follows that $x_{s-1,s}=1$.
  Thus, we obtain from constraints~\eqref{ip:ext-tree-7} and~\eqref{ip:ext-tree-8} that $z_{sk} = 1$.
\end{proof}
\begin{lemma}\label{lem:tree-poly-corr-merge}
  Let $(x^\ast,y^\ast,a^\ast,w^\ast,z^\ast,R^\ast)$ be a feasible integer solution to~\eqref{ip:ext-tree} and let $P$ be an $s$-$t$-path in $\fxystar$ with $z^\ast_{sk} = z^\ast_{t\ell} = 1$
  for some $k < \ell$, $s\not=t$.
  Then $w^\ast_{k\ell}=1$.
\end{lemma}
\begin{proof}
  Consider the unique edge $\{j,t\}$ on $P$ that is incident to $t$;
  in particular, let $j \in V$ be on $P$.
  Lemma~\ref{lem:tree-poly-corr-path} yields that $z_{jk} = z_{tk} = 1$.
  Thus, by constraints~\eqref{ip:ext-tree-5} and~\eqref{ip:ext-tree-6}, we have $w_{k\ell} \geq x_{jt} = 1$.
\end{proof}
\begin{definition}
  Let $(x^\ast,y^\ast,a^\ast,w^\ast,z^\ast,R^\ast)$ be a feasible integer solution to~\eqref{ip:ext-tree}.
  We say that the root $r^k$ is active if and only if $a^\ast_k=1$, for all $k \in [K]$.
\end{definition}
The following lemma follows immediately from constraint~\eqref{ip:ext-tree-3}.
\begin{lemma}\label{lem:tree-poly-corr-repr}
Let $(x^\ast,y^\ast,a^\ast,w^\ast,z^\ast,R^\ast)$ be a feasible integer solution to~\eqref{ip:ext-tree}.
Then exactly $R^\ast$ out of $r_1,\dots,r^K$ are active.\qed
\end{lemma}
\begin{lemma}\label{lem:tree-poly-corr-active}
Let $(x^\ast,y^\ast,a^\ast,w^\ast,z^\ast,R^\ast)$ be a feasible integer solution to~\eqref{ip:ext-tree}
and let $F$ be a connected component of $\fxystar$.
Then there is exactly one $k \in [K]$ such that $r^k \in F$ and $r^k$ is active.
\end{lemma}
\begin{proof}
Assume that $r^k, r^\ell \in F$ and that $a^\ast_k=a^\ast_\ell=1$.
In particular, there is an $r^k$-$r^\ell$-path~$P$ in~$F$. 
By the fixing~\eqref{ip:ext-tree-12}, we have $z^\ast_{{r^k},k}=z^\ast_{r^\ell,\ell}=1$ and Lemma~\ref{lem:tree-poly-corr-merge} yields that $w_{k\ell} = 1$.  
This contradicts constraint~\eqref{ip:ext-tree-4} and thus, no connected component of $\fxystar$ can contain two active roots.
On the other hand, by Lemma~\ref{lem:tree-poly-corr-repr} there are exactly as many active representatives as there are connected components in $\fxystar$.
Applying the pidgeon-hole principle tells us that each connected component of $\fxystar$ must contain exactly one active root.
\end{proof}
\begin{lemma}
Let $(x^\ast,y^\ast,a^\ast,w^\ast,z^\ast,R^\ast)$ be a feasible integer solution to~\eqref{ip:ext-tree} and let $F$ be a connected component of $\fxystar$.
Suppose that $F$ contains a terminal node $t \in T^k$ for some $k \in [K]$.
Then $t$ also contains $r^k$.
\end{lemma}
\begin{proof}
Assume that $r^k \not\in F$, but instead $r^k \in F' \not= F$.
Then, by Lemma~\ref{lem:tree-poly-corr-active}, there are active roots $r^\ell \in F$ and $r_m \in F'$. 
%In particular, the variables $r^{\ast}_{\ell}$ and $r^{\ast}_m$ are set to 1.
Also, by construction, there is a $t$-$r_{\ell}$-path in $F$ and there is a $r^k$-$r_m$-path in $F'$.
%Suppose now that $k < \ell$ (the other case can be obtained by reversing the roles of $k$ and $\ell$), then.
By Lemma~\ref{lem:tree-poly-corr-merge} we know that either $w^{\ast}_{k\ell}=1$ or $w^\ast_{\ell k} = 1$, depending on whether $\ell < k$ or $k < \ell$; and either $w^\ast_{km}=1$ or $w^\ast_{mk}=1$, depending on whether $m < k$ or $k < m$.
Thus, the transitivity constraints~\eqref{ip:ext-tree-9}--\eqref{ip:ext-tree-11} ensure that either $w^\ast_{m\ell} =  1$ or $w^\ast_{\ell m} = 1$.
In either case, this is a contradiction to constraint~\eqref{ip:ext-tree-4} because both $a^\ast_\ell=1$ and $a^\ast_m = 1$.
\end{proof}
\begin{corollary}
If $(x^\ast,y^\ast,a^\ast,w^\ast,z^\ast,R^\ast)$ is a feasible integer solution to~\eqref{ip:ext-tree},
then $\fxystar$ is a feasible Steiner Forest with the same cost.
If $F'=(V'_F, E'_F)$ is a feasible Steiner Forest, then there exists a feasible solution $(x^\ast,y^\ast,a^\ast,w^\ast,z^\ast,R^\ast)$ to~\eqref{ip:ext-tree} with the same cost.
\end{corollary}

\subsection{A Flow-Based Directed Formulation}
The following three bi-directed formulations are based on the orientation argument by Magnanti and Raghavan~\cite{MR2005} but turn out to be significantly smaller.
Consider a feasible Steiner Forest~$F$ with connected components~$F^1,\dots,F^Q$, $Q \leq K$.
For each~$q \in [Q]$, we denote the lowest index of any root in component~$F^q$ by $\rho(q)$.
We then repeat the argument from~\cite{MR2005}: 
The feasible forest~$F$ can be directed such that each connected component~$F^q$ is an arborescence rooted at $r^{\rho(q)}$, i.e., all arcs point away from~$r^{\rho(q)}$:
If~$t \not= r^{\rho(q)}$ is a terminal in~$F^q$ and~$S \subseteq V$ is a cut with~$r^{\rho(q)} \in S, t \not\in S$, then~$\delta^+_F(S)$ cannot be not empty.  

The problem is again that we cannot know beforehand what the $F^1,\dots,F^Q$ will be and thus we cannot know the lowest index roots, either.
Instead, we want an IP-Formulation that finds an optimum \emph{directed} Steiner Forest on the bidirected graph induced by $G$ and makes sure that this forest satisfies the Magnanti-Raghavan condition.
Any feasible (undirected) Steiner Forest will then correspond to a feasible solution of our IP-formulation and vice-versa; thus we can use the IP-formulation to find optimum undirected Steiner Forests.

Our first formulation is a flow-based one, and it uses the same principles as~\eqref{ip:mr}:
For each $\ell \in [K]$, all non-root terminals $t \in T^\ell$ must receive one unit of flow in total. 
The flow may be sent from any root $r^k$ with $k\leq \ell$, and in a fractional solution, multiple roots can send flow to $t$ at the same time.
However, if $z^{k\ell}$ is the combined value of flow sent from $r^k$ to the terminals in $T^\ell$, then there must be a flow of value $z^{k\ell}$ from $r^k$ to $r^\ell$.
 
To describe the formulation in detail, we need some additional notation.
Let $T_r^k := T^k\setminus\{r^k\}$. 
Moreover, let $\terms^{k\ldots K} := \bigcup_{\ell\in\{k, \ldots, K\}} T^\ell, \forall k\in[K]$, denote the union of terminal sets $T^k, \ldots, T^K$ and 
$\terms_r^{k\ldots K} := \terms^{k\ldots K}\setminus\{r^k\}$ the same set without the $k$th root node (all other root nodes are still included). 

The formulation has two flow variables $f^{kt}_{ij}, f^{kt}_{ji}$ for each edge $\{i,j\} \in E$, each $k \in [K]$ and all terminals $t \in \terms^{k\ldots K}$.
Aside from these standard flow variables, it has variables $z_{k\ell}\in \{0,1\}$  for each pair of terminal sets $k,\ell \in[K]$, with $k \leq \ell$. 
The model ensures that $z_{k\ell} = 1$ iff $r^\ell$ (and all vertices in $T^\ell$) are contained in the arborescence rooted at $r^k$;
$z_{kk} = 1$ implies that $r^k$ is a root node of an arborescence itself. 
We call $r^k$ with $z_{kk} = 1$ a {\em parent node} and say that $r^k$ is the {\em parent of} $r^\ell$ if $z_{k\ell} = 1$; in this case $r^\ell$ is the {\em child of} $r^k$. 
Each root node is either a parent or a child node and $r^1$ is always a parent node -- we set $z_{11} = 1$ (this is implied by the following constraints \eqref{ip:ext-dir-flow-z-sum-1}).
The separate set of variables $x^k \in \{0,1\}^E, k\in[K]$, for each terminal set models the edge capacities: The $k$th set of variables contains exactly those edges of the tree rooted at $r^k$.  
Finally, a variable $x_{ij}$ for each edge models the decision whether edge $\{i,j\}$ is included in the forest.
% mapping of Bernds to Daniels notations: 
% T^{k\ldots K} = \terms^{k\ldots K}
% T_r^{k\ldots K} = \terms_r^{k\ldots K}
We have
\begin{align}
  \min\Bigsetm{c^\transp x}{(x, f, z) \in \LP^{edf}\ \tandinteger} \tag{IPedf} \label{ip:ext-dir-flow}
\end{align}
where
\begin{subequations}
\begin{align}
    \LP^{edf} := \Bigl\{ (x, f, z) \mathrel{}\Bigm\lvert 
     	x_{ij}            &\geq  \sum_{k\in[K]} x_{ij}^k  &&
                                         \begin{aligned}
                                           &\tforall\ \{i,j\} \in E
                                         \end{aligned}\label{ip:ext-dir-flow-capacity}\\
         \sum_{\ell=1}^k z_{\ell k} &= 1 && \begin{aligned}
                                          &\tforall\ k \in [K]\\
                                        \end{aligned}\label{ip:ext-dir-flow-z-sum-1}\\
	z_{kk}  &\ge z_{k\ell} && \begin{aligned}
                                         &\tforall\ k \in[K]\setminus\{1, K\} \\
                                         &\tandall\ \ell \geq k+1 
                                        \end{aligned}\label{ip:ext-dir-flow-z}\\
	 f_{ij}^{ks} + f_{ji}^{kt} &\leq x_{ij}^k && \begin{aligned}
	 	&\tforall\ k\in[K] \\
		&\tandall\ e=\{i,j\}\in E \\
		&\tandall\ s, t \in \terms_r^{k\ldots K} 
		\end{aligned}
		\label{ip:ext-dir-flow-edges-capacity}\\
   f^{kt}(\delta^-(i)) - f^{kt}(\delta^+(i)) &=
    \begin{cases}
         -z_{k\ell},&\text{if $i = r^{k}$}\\
        z_{k\ell},&\text{if $i = t$}\\
         0,&\text{otherwise}
    \end{cases}
    && \begin{aligned} 
	 &\tforall\ k\in[K]\\ 
	&\tandall\ t \in \terms_r^{k\ldots K} \\
	&\text{with }\tau(t) = \ell \\
    	&\tandall\ i \in V\\ 
	\end{aligned}
	\label{ip:ext-dir-flow-flow-conservation}\\
        x_{ij}                 &\in [0,1] && \tforall\ \{i,j\}\in E \\
        x_{ij}^k	      &\in [0,1] &&  \begin{aligned} &\tforall\ k\in [K]\\ &\tandall\ \tforall\ \{i,j\}\in E \end{aligned}\\
        f_{ij}^{kt}, f_{ji}^{kt}		&\in [0,1] && \begin{aligned} &\tforall\ k\in [K]\\ &\tandall\ t \in \terms_r^{k\ldots K} \\ &\tandall\ \{i,j\}\in E \end{aligned} \\
        z_{k\ell}           &\in [0,1] && \begin{aligned} &\tforall\ k\in [K]\\ &\tandall\ \ell \geq k\end{aligned}\, \Bigr\}.
\end{align}
\end{subequations}

The constraints \eqref{ip:ext-dir-flow-z-sum-1} and \eqref{ip:ext-dir-flow-z} imply a valid assignment of the $z$-variables by modeling a flat hierarchy between the root nodes:
a root node $r^k$ is either a parent or a child of exactly one other root node. 
\eqref{ip:ext-dir-flow-z-sum-1} states that every root $r^k, k\in[K]$, has to be a parent, i.e., a root node of an arborescence ($z_{kk} = 1$), or it has to be a child and it has to be  contained in another arborescence ($\exists i < k\colon z_{ik} = 1$). 
\eqref{ip:ext-dir-flow-z} states that if a root node $r^i$ is a child of another root node $r^k$ then $r^k$ has to be a parent node. 

Conditions \eqref{ip:ext-dir-flow-edges-capacity} and \eqref{ip:ext-dir-flow-flow-conservation} model a valid flow. 
Thereby, a flow of value $z_{k\ell}$ is send from root node $r^k$ to a terminal $t\in \terms_r^{k\ldots K}$ with $\tau(t)=\ell$. 
Hence, if $r^k$ is the parent of set $\ell$ each terminal in $T^\ell$ is connected to $r^k$.  

The constraints \eqref{ip:ext-dir-flow-edges-capacity} ensure the correct assignment of flow- and edge-variables. 
Here, the constraints affect each tree separately. 
The fact that an optimum solutions to the SFP consists of a disjoint set of trees is represented by the sum in constraints \eqref{ip:ext-dir-flow-capacity}: 
hence, any edge used in any tree needs to be payed for.

\begin{lemma}
\label{lemma:correctness:ext-dir-flow}
\eqref{ip:ext-dir-flow} models the Steiner Forest problem correctly. 
\end{lemma}
\begin{proof}
Let $\tilde E \subseteq E$ be an optimal solution to the SFP. 
Start with $\tilde z := \mathbf{0}$. 
Now, for each connected component $\mathcal{C}$ in $G[\tilde E]$ set $\tilde z_{ii} = 1$ if $r^i$ is the root node with lowest index contained in $\mathcal{C}$ and
for all other root nodes $r^j\in\mathcal{C}, j\not= i$, set $\tilde z_{ij} = 1$.  
Notice that $\tilde z$ satisfies \eqref{ip:ext-dir-flow-z-sum-1} and \eqref{ip:ext-dir-flow-z} and that each terminal has exactly one assigned parent node. 
After fixing the $z$ variables the remaining part of the model describes a union of disjoint Steiner trees, one for each connected component. 
First, the component with parent node $r^k$ is represented by edge variables $\tilde x^k$. 
Second, $\tilde E$ can be oriented such that each connected component is an arborescence rooted at its parent node. 
Then, the arcs of the arborescences can be used for constructing flows from each parent node $r^i$ to each terminal $t\in T_r^i$ or $t\in T^j$ with $j>i$ and $z_{ij} = 1$. 
Since the connected components are disjoint constraint  \eqref{ip:ext-dir-flow-capacity} is satisfied.
Overall, the constructed solution is feasible for  \eqref{ip:ext-dir-flow} and has the same objective value.

An optimum solution $(\tilde x, \tilde f, \tilde z)$ to \eqref{ip:ext-dir-flow}  implies a hierarchy of the terminal sets with parent and child sets. 
Thereby, every set has exactly one assigned parent; in particular, every terminal of a set has the same parent.  
Hence, due to \eqref{ip:ext-dir-flow-flow-conservation} there exists a flow of one unit from each terminal to the assigned parent such that every terminal is connected. 
Constraints \eqref{ip:ext-dir-flow-capacity} and \eqref{ip:ext-dir-flow-edges-capacity} collect the used edges and hence, $\tilde E := \{\{i,j\}\in E \mid \tilde x_{ij} = 1\}$ is a feasible solution to the SFP with the same cost. 
\end{proof}

Let  $\Proj_x(\LP^{edf})$ and $\Proj_x(\LP^{df})$
denote the linear projections of 
$\LP^{edf}$ and $\LP^{df}$, respectively, 
into the undirected $x$ variable space.

\begin{lemma}
\label{lemma:df:vs:edf}
$
\Proj_x(\LP^{df})
\supsetneq
\Proj_x(\LP^{edf})
$, i.e., 
the extended directed flow-based formulation is stronger than the directed flow-based formulation. 
\end{lemma}
\begin{proof}
Let $(\tilde x, \tilde f, \tilde z) \in \LP^{edf}$. 
For better overview we divide the proof into several parts.
Parts (A)--(D) show that $\Proj_x(\LP^{edf}) \subseteq \Proj_x(\LP^{df})$ and 
(E) gives an example where the strict inequality holds. 

\subparagraph{A. Flows are acyclic.} 
W.l.o.g. we can assume that any flow $\tilde f^{kt}, \tforall\ k\in[K] \tandall\ t \in \terms_r^{k\ldots K}$, is free of cycles and it satisfies $\tilde f_{ij}^{kt}=0 \vee \tilde f_{ji}^{kt} = 0, \forall \{i,j\}\in E$.

\subparagraph{B. Reverse flow.} 
We first introduce additional flow variables $\check f^{k r^\ell}$, $\forall \ell\in\{1, \ldots, K-1\}, \forall k\in \{\ell + 1, \dots, K\}$, i.e., $k > \ell$.  
Notice these flow variables do not exist since we have only flow variables $f^{kt}$ for a set $k$ and terminal $t\in \terms_r^{k\ldots K} $, i.e., $\tau(t)$ $\geq k$.   
The values of the new variables are set such that the flow from $r^\ell$ to $r^k$ is simply reversed: $\forall (i,j)\in A\colon \check f_{ij}^{k r^\ell} := \tilde f_{ji}^{\ell r^k}$.

\begin{figure}[t]
\begin{subfigure}[t]{0.4\textwidth}
  \centering
  \begin{tikzpicture}
      \useasboundingbox (-0.5,-0.7) rectangle (2.3,2.7);
  \begin{scope}[scale=1.75]
    \draw (0,1) 	node[rB,label={[mylabel]120:$r^k$}] (vrk) {};
    \draw (1,1) 	node[rA,label={[mylabel]30:$r^\ell$}] (vrl) {};
    \draw (0.5,0) node[tB,label={[mylabel]300:$t$}] (vt) {};
  \end{scope}
  \begin{scope}[myedge, sloped]
  \draw[<-,decoration={zigzag,segment length=10, amplitude=1.5,post=lineto, post length=4pt, pre length=4pt}] (vrk) edge[decorate] node[above] {$\tilde f^{\ell, r^k}$}  (vrl);
  \draw[->,decoration={zigzag,segment length=10, amplitude=1.5,post=lineto, post length=4pt}] (vrk) edge[decorate] node[below] {$\tilde f^{k, t}$}  (vt);
  \draw[->,decoration={zigzag,segment length=10, amplitude=1.5,post=lineto, post length=4pt}] (vrl) edge[decorate] node[below] {$\tilde f^{\ell, t}$}  (vt);
  \end{scope}
\end{tikzpicture}
\caption{}
\end{subfigure}\qquad
\begin{subfigure}[t]{0.4\textwidth}
  \centering
  \begin{tikzpicture}
      \useasboundingbox (-0.7,-0.5) rectangle (2.9,2.7);
  \begin{scope}[scale=1.75]
    \draw (0,1) 	node[rB,label={[mylabel]120:$r^k$}] (vrk) {};
    \draw (1,1) 	node[rA,label={[mylabel]30:$r^\ell$}] (vrl) {};
    \draw (0.5,0) node[tB,label={[mylabel]300:$t$}] (vt) {};
  \end{scope}
  \begin{scope}[myedge, sloped]
  \draw[->,decoration={zigzag,segment length=10, amplitude=1.5,post=lineto, post length=4pt}] (vrk) edge[decorate] node[above] {$\check f^{k, r^\ell}$}  (vrl);
  \draw[->,decoration={zigzag,segment length=10, amplitude=1.5,post=lineto, post length=4pt}] (vrk) edge[decorate] node[below] {$\tilde f^{k, t}$}  (vt);
  \draw[->,decoration={zigzag,segment length=10, amplitude=1.5,post=lineto, post length=4pt}] (vrl) edge[decorate] node[below] {$\tilde f^{\ell, t}$}  (vt);
  \draw[->, dashed] (vrk) .. controls (0.2,3.5) and (6,1.5) .. node[above] {$\bar f^{k, \ell, t}$} (vt);
  \end{scope}
\end{tikzpicture}
\caption{}
\end{subfigure}\qquad
\caption{\label{figure:df:vs:edf}
Schematic view on the involved flows in the proof of Lemma \ref{lemma:df:vs:edf}. 
$r^k$ and $r^\ell$ are root nodes for sets $T^k$ and $T^\ell$, with $\ell < k$ and $t\in T_r^k$. 
(a) The original flows. 
(b) The reverse flow $\check f^{k r^\ell}$ from $r^k$ to $r^\ell$, cf.\ part {\em B} in the proof, and the combined flow $\bar f^{k \ell t}$ from $r^k$ to $t$ over $r^\ell$, cf.\ part {\em C}. 
}
\end{figure}
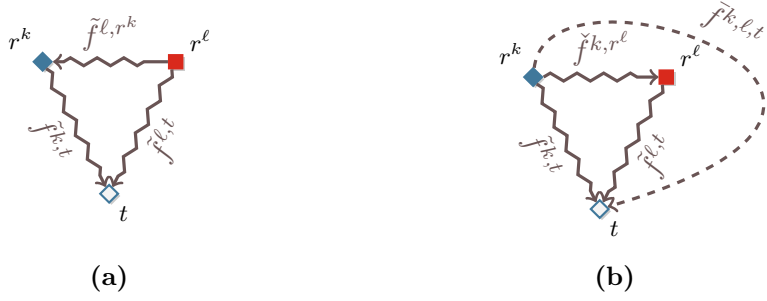

\subparagraph{C. Flow from $r^k$ to $t$ {\em over} $r^\ell$.}
Now, we construct a flow $\bar f^{k\ell t}$ for a set $k\in[K]\setminus\{1\}$, a set $\ell \in \{1, \ldots, k-1\}$, and a terminal $t\in T_r^k$.  
This flow will send $\tilde z_{\ell k}$ from $r^k$ to $t$ (over~$r^\ell$) by using the reverse flow from $r^\ell$ to $r^k$, i.e., $\bar f^{k\ell t} := \tilde f^{\ell t} + \check f^{k r^\ell}$.

\subparagraph{C.1. Feasibility and value.}
We show that $\bar f^{k\ell t}$ is a feasible flow from $r^k$ to $t$ with value $\tilde z_{\ell k}$, $\forall k\in [K]\setminus\{1\}, \forall t\in T_r^k, \forall \ell \in \{1, \ldots, k-1\}$. 
Let $i\in V$. 
We have:
	\begin{align}
	&\bar f^{k\ell t}(\delta^-(i)) - \bar f^{k\ell t}(\delta^+(i))  \nonumber\\
	=\, & \tilde f^{\ell t}(\delta^-(i)) + \check f^{k r^\ell}(\delta^-(i)) - \tilde f^{\ell t}(\delta^+(i)) - \check f^{k r^\ell}(\delta^+(i)). \nonumber
	\end{align}

\begin{itemize}
\item[] {\em Case} ``$i=r^k$'': $\tilde f^{\ell t}(\delta^-(r^k)) - \tilde f^{\ell t}(\delta^+(r^k)) = 0$ since $r^k$ is an internal node under flow $\tilde f^{\ell t}$. 
					Moreover, $\check f^{k r^\ell}(\delta^-(r^k)) - \check f^{k r^\ell}(\delta^+(r^k)) = -\tilde z_{\ell k}$ (the reverse flow).

\item[] {\em Case} ``$i=t$'':  Similar arguments: $\check f^{k r^\ell}(\delta^-(t)) - \check f^{k r^\ell}(\delta^+(t)) = 0$ since $t$ is an internal node under $\check f^{k r^\ell}$
				and  $\tilde f^{\ell t}(\delta^-(t)) - \tilde f^{\ell t}(\delta^+(t)) = \tilde z_{\ell k}$. 

\item[] {\em Case} ``$i=r^\ell$'': $\tilde f^{\ell t}(\delta^-(r^\ell)) - \tilde f^{\ell t}(\delta^+(r^\ell)) = -\tilde z_{\ell k}$ and $\check f^{k r^\ell}(\delta^-(r^\ell)) - \check f^{k r^\ell}(\delta^+(r^\ell)) = \tilde z_{\ell k}$. 
				Hence, the sum is 0. 
				
\item[] {\em Otherwise} :	Since $\tilde f^{\ell t}$ and $\check f^{k r^\ell}$ are flows the sum is 0. 
\end{itemize}

Hence, $\bar f^{k\ell t}$ is a feasible flow from $r^k$ to $t$ with value $\tilde z_{\ell k}$. 

\subparagraph{C.2. Acyclic $\bar f^{k\ell t}$.} 
Again, we assume w.l.o.g.\ that $\bar f^{k\ell t}_{ij}$ is acyclic, i.e, $\bar f_{ij}^{k\ell t}=0 \vee \bar f_{ji}^{k\ell t} = 0, \forall \{i,j\}\in E$. 

\subparagraph{C.3. Capacity: $\bar f^{k\ell s}_{ij} + \bar f^{k\ell t}_{ji} \leq \tilde x_{ij}^\ell$.} 
Now, for any $k \in [K]\setminus\{1\}$ and any $\ell \in \{1,\ldots, k-1\}$, consider two terminals $s, t\in T_r^k$ from the same terminal set, and an edge $\{i,j\}\in E$ with the two related arcs $a_1\in\{(i,j), (j,i)\}$ and  the reverse arc $a_2$.
We argue that $\bar f^{k\ell s}_{a_1} + \bar f^{k\ell t}_{a_2} \leq \tilde x_{ij}^\ell$.

If one flow is zero the inequality holds: 
E.g., if 
$\bar f^{k\ell t}_{a_2} = 0 
$
we have: 
$
\bar f^{k\ell s}_{a_1} + \bar f^{k\ell t}_{a_2}
=
\bar f_{a_1}^{k \ell s}
=
\tilde f_{a_1}^{\ell s} + \check f_{a_2}^{\ell r^k} 
\leq 
\tilde x_{ij}^\ell
$. 
The last inequality  is true due to constraint \eqref{ip:ext-dir-flow-edges-capacity}. 
The part with $\bar f^{k\ell s}_{a_1} = 0$ works analogously.

Otherwise, if both parts are $> 0$ we have:
$\bar f^{k\ell s}_{a_1} + \bar f^{k\ell t}_{a_2}
= \tilde f_{a_1}^{\ell s} + \check f_{a_1}^{k r^\ell} - \tilde f_{a_2}^{\ell s} - \check f_{a_2}^{k r^\ell} + 
   \tilde f_{a_2}^{\ell t} + \check f_{a_2}^{k r^\ell} - \tilde f_{a_1}^{\ell t} - \check f_{a_1}^{k r^\ell}
= \tilde f_{a_1}^{\ell s} - \tilde f_{a_2}^{\ell s} + \tilde f_{a_2}^{\ell t} - \tilde f_{a_1}^{\ell t}
\leq \tilde x_{ij}^\ell
$, again by constraint \eqref{ip:ext-dir-flow-edges-capacity}.  

\subparagraph{D. Solution to $\LP^{df}$.}  
Due to the previous discussion we are now able to construct a solution $(\hat x, \hat f) \in  \LP^{df}$ with the same objective value. 

\subparagraph{D.1. Variable assignment.} 
We use the same values for the undirected edges by assigning $\hat x := \tilde x$. 
Trivially, $\hat x\in[0,1]^{|E|}$. 

The flow variables $\bar f^{t}, \forall t\in\nonrootterms$, with $k=\tau(t)$, are assigned the following values: $\hat f^{t} := \tilde f^{kt} + \sum_{\ell\in\{1, \ldots, k-1\}} \bar f^{k\ell t}$. 
Obviously, it holds $\hat f^{t} \geq 0$; the upper bound of 1 follows from {\em D.3}. 

\subparagraph{D.2. Flow conservation and flow value 1.}
Consider a terminal $t\in\nonrootterms$ with $k=\tau(t)$ and a vertex $i\in V$. 
By inserting the definition we have: 
	\begin{align}
	&\hat f^{t}(\delta^-(i)) - \hat f^{t}(\delta^+(i)) \nonumber\\
	=\, & \tilde f^{kt}(\delta^-(i)) + \sum_{\ell\in\{1, \ldots, k-1\}} \bar f^{k\ell t}(\delta^-(i)) - \tilde f^{kt}(\delta^+(i)) - \sum_{\ell\in\{1, \ldots, k-1\}} \bar f^{k\ell t}(\delta^+(i)) \nonumber
	\end{align}

\begin{itemize}
\item[] {\em Case}  ``$i=r^k$'':  	$\tilde f^{kt}(\delta^-(i)) - \tilde f^{kt}(\delta^+(i)) = - \tilde z_{kk}$ and for each $\ell < k$ it holds $\bar f^{k\ell t}(\delta^-(i)) - \bar f^{k\ell t}(\delta^+(i)) = - \tilde z_{\ell k}$ (due to {\em C.1}). Overall we get $- \tilde z_{kk} + \sum_{\ell < k} - \tilde z_{\ell k} = - 1$ (due to constraint \eqref{ip:ext-dir-flow-z-sum-1}). 

\item[] {\em Case}  ``$i=t$'':  		Analogously, $\tilde f^{kt}(\delta^-(i)) - \tilde f^{kt}(\delta^+(i)) = \tilde z_{kk}$ and for each $\ell < k$ it holds $\bar f^{k\ell t}(\delta^-(i)) - \bar f^{k\ell t}(\delta^+(i)) = \tilde z_{\ell k}$ (due to {\em C.1}), and overall we have $\tilde z_{kk} + \sum_{\ell < k} \tilde z_{\ell k} = 1$ (due to constraint \eqref{ip:ext-dir-flow-z-sum-1}). 	 
				
\item[] {\em Otherwise}: 		Since $\tilde f^{kt}$ and $\bar f^{k\ell t}(\delta^-(i)), \forall \ell < k$, are flows (see {\em C.1}) the sum is~0.  
\end{itemize}

We conclude that $\hat f^{t}$ is a flow from $r^k$ to each terminal $t\in T_r^k, \forall k\in[K]$, with value~1. 

\subparagraph{D.3. $\hat x_{ij} \ge \hat f_{ij}^{s} + \hat f_{ji}^{t}$.} 
Last but not least, we need to show that constraints \eqref{ip:dir-flow-2} are satisfied. 
Let $\{i,j\}\in E$, $k\in[K]$, and $s, t\in T_r^k$. 
	\begin{align}
	\hat f_{ij}^{s} + \hat f_{ji}^{t} 	\stackrel{\hphantom{(5.55)}}{\leq} & \tilde f_{ij}^{ks} + \sum_{\ell\in\{1, \ldots, k-1\}} \bar f_{ij}^{k\ell s} + \tilde f_{ji}^{kt} + \sum_{\ell\in\{1, \ldots, k-1\}} \bar f_{ji}^{k\ell t} \nonumber \\
								\stackrel{\eqref{ip:ext-dir-flow-edges-capacity}}{\leq} & \tilde x_{ij}^k +  \sum_{\ell\in\{1, \ldots, k-1\}} \left( \bar f_{ij}^{k\ell s} + \bar f_{ji}^{k\ell t} \right) \nonumber\\
								\stackrel{\mathit{C.3}}{\leq} & \tilde x_{ij}^k + \sum_{\ell\in\{1, \ldots, k-1\}} \tilde x_{ij}^\ell 
														\leq \sum_{k\in\{1, \ldots, K\}} \tilde x_{ij}^k  
														\stackrel{\eqref{ip:ext-dir-flow-capacity}}{\leq} \tilde x_{ij} = \hat x_{ij}\nonumber
	\end{align}

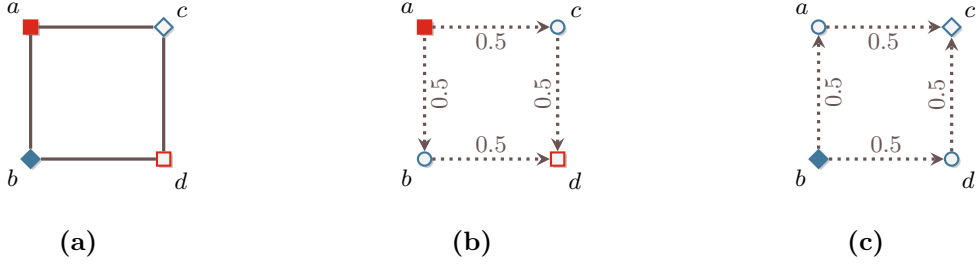
\begin{figure}[t]
\begin{subfigure}[t]{0.3\textwidth}
  \centering
  \begin{tikzpicture}
      \useasboundingbox (-0.7,-0.7) rectangle (2,2);
  \begin{scope}[scale=1.75]
    \draw (0,0) node[rB,label={[mylabel]210:$b$}] (v1) {};
    \draw (0,1) node[rA,label={[mylabel]120:$a$}] (v2) {};
    \draw (1,1) node[tB,label={[mylabel]30:$c$}] (v3) {};
    \draw (1,0) node[tA,label={[mylabel]300:$d$}] (v4) {};
  \end{scope}
  \begin{scope}[myedge, mylabel, sloped, above]
  \draw (v1) -- node{} (v2);
  \draw (v1) -- node{} (v4);
  \draw (v2) -- node{} (v3);
  \draw (v3) -- node{} (v4);
  \end{scope}
\end{tikzpicture}
\caption{}
\end{subfigure}\qquad
\begin{subfigure}[t]{0.3\textwidth}
  \centering
  \begin{tikzpicture}
   \useasboundingbox (-0.7,-0.7) rectangle (2,2);
  \begin{scope}[scale=1.75]
    \draw (0,0) node[mybluenode,label={[mylabel]210:$b$}] (v1) {};
    \draw (0,1) node[rA,label={[mylabel]120:$a$}] (v2) {};
    \draw (1,1) node[mybluenode,label={[mylabel]30:$c$}] (v3) {};
    \draw (1,0) node[tA,label={[mylabel]300:$d$}] (v4) {};
  \end{scope}
  \begin{scope}[myedge, ->, >=stealth, dotted, sloped, below]
  \draw[<-] (v1) to node[mylabel]{0.5} (v2);
  \draw (v2) to node[mylabel]{0.5} (v3);
  \draw[<-] (v4) to node[mylabel,above]{0.5} (v3);
  \draw (v1) to node[mylabel,above]{0.5} (v4);
  \end{scope}
\end{tikzpicture}
  \caption{}
\end{subfigure}\qquad
\begin{subfigure}[t]{0.3\textwidth}
  \centering
  \begin{tikzpicture}
      \useasboundingbox (-0.7,-0.7) rectangle (2,2);
  \begin{scope}[scale=1.75]
    \draw (0,0) node[rB,label={[mylabel]210:$b$}] (v1) {};
    \draw (0,1) node[mybluenode,label={[mylabel]120:$a$}] (v2) {};
    \draw (1,1) node[tB,label={[mylabel]30:$c$}] (v3) {};
    \draw (1,0) node[mybluenode,label={[mylabel]300:$d$}] (v4) {};
  \end{scope}
  \begin{scope}[myedge,->, >=stealth, dotted, sloped, below]
  \draw (v1) to node[mylabel,below]{0.5} (v2);
  \draw (v2) to node[mylabel]{0.5} (v3);
  \draw (v4) to node[mylabel,above]{0.5} (v3);
  \draw (v1) to node[mylabel,above]{0.5} (v4);
  \end{scope}
\end{tikzpicture}
\caption{}
\end{subfigure}
\caption{\label{figure:df:vs:edf}
An instance where the LP relaxation of the extended directed flow formulation gives a better bound than the directed flow formulation, cf.\ part (E) in the proof of Lemma \ref{lemma:df:vs:edf}. 
(a) depicts the input graph and  
(b) and (c) give valid flows for sets 1 and 2 (dashed arcs route flow of 0.5).
}
\end{figure}

\subparagraph{E. Example for strict inequality.} 
Figure \ref{figure:df:vs:edf} gives an example with $x\in \Proj_x(\LP^{df})$ but $x\not\in \Proj_x(\LP^{edf})$. 
The instance has unit edge costs and the two terminal sets $T^1 = \{a, d\}$ and $T^2 = \{b, c\}$ with $r^1 = a, r^2 = b$. 
The optimum solution to $\LP^{df}$ sets $x_{ij} := 0.5, \forall \{i,j\}\in E$, and the flows are given by 
Figure (b) and (c) with dashed arcs routing a flow value of value $0.5$. 
Hence, the optimum solution value of $\LP^{df}$ is 2.

On the other hand, this solution is not valid for model $\LP^{edf}$. 
A value of 0.5 for each edge implies a flow for the first terminal set as depicted in Figure (b). 
Then, it is not possible to route any flow for the second set (from node $b$ to $c$) without increasing the $x$ variables. 
Hence, it has to hold $z_{12} = 1$. 
However, sending a flow with value 1 from $a$ to nodes $b$ and $c$ while using the same arcs as in (b) is not possible. 
It is easy to see that the optimum solution to the LP relaxation of $\LP^{edf}$ has a value of  3 by picking any three edges. 
\end{proof}

\subsection{A Cut-Based Directed Formulation}
Let us now derive two directed cut-based formulations from~\eqref{ip:ext-dir-flow}.
The advantage of the first formulation is that it has few variables: 
We need two variables~$y_{ij}, y_{ji}$ and a variable~$x_{ij}$ for each edge~$\{i,j\} \in E$.
Additionally, for all~$k \in [K]$ and all $\ell \geq k$, we have a decision variable~$z_{k\ell}$ that tells us whether the terminals in $T^\ell$ should be connected to the root $r^k$, as before.
\begin{align}
  \min\Bigsetm{c^\transp x}{(x,y,z) \in \LP^{edc}\ \tandinteger} \tag{IPedc} \label{ip:ext-dir-cut}
\end{align}
where
\begin{subequations}
\begin{align}
    \LP^{edc} := \Bigl\{ (x,y,z) \mathrel{}\Bigm\lvert 
     y(\delta^+(S))            &\geq \srsum{\substack{\ell\leq k:\\ r^\ell \in S}} z_{\ell k} &&
                                         \begin{aligned}
                                           &\tforall\ k \in [K]\\
					  &\tandall\ S\subseteq V\colon T^k\cap S\not=T^k
                                         \end{aligned}\label{ip:ext-dir-cut-1}\\
         \sum_{\ell=1}^k z_{\ell k} &= 1 && \tforall\ k \in [K] \label{ip:ext-dir-cut-3}\\
        z_{kk} &\geq z_{k\ell} &&\begin{aligned}
                                 &\tforall\ k \in [K]\setminus\{1,K\}\\
                                 &\tandall\ \ell \geq k+1
                               \end{aligned}\label{ip:ext-dir-cut-z}\\
         y_{ij} + y_{ji}           &\leq x_{ij} && \tforall\ \{i,j\} \in E\label{ip:ext-dir-cut-4}\\
         y_{ij}, y_{ji}, x_{ij}    &\in [0,1] && \tforall\ \{i,j\} \in E\\
         z_{k\ell}                 &\in [0,1] && \begin{aligned} &\tforall\ k\in [K]\\ &\tandall\ \ell \geq k\end{aligned}\,\Bigr\}.
\end{align}
\end{subequations}

Using the directing procedure from~\cite{MR2005} we can construct a feasible solution to~\eqref{ip:ext-dir-cut} from any Steiner Forest.
\begin{lemma}
  Let $F:=(V_F, E_F)$ be a feasible Steiner Forest for $(G,\terms, c)$ with cost $C$.
  Then, there exists an integer feasible solution $(x^\ast, y^\ast, z^\ast)$ for~\eqref{ip:ext-dir-cut} with value $C$.
\end{lemma}
\begin{proof}
  To construct a feasible solution~$(x^\ast, y^\ast, z^\ast)$, we first orient the edges in~$F$ with the procedure from~\cite{MR2005} and denote the oriented edge set by~$A$.
  Suppose that~$F^1,\dots,F^Q$ are the connected components of~$F$ and let~$\rho(q)$ be the lowest index of any root node in~$F^q$, for~$q \in [Q]$.
  Then, for all~$\{i,j\}$, we set~$y^\ast_{ij} = 1$ or~$y^\ast_{ji} =1$ if we have~$(i,j) \in A$ or~$(j,i) \in A$, respectively. We set~$x^\ast_{ij} = 1$ in both cases. 
  Otherwise, we let~$y_{ij} = y_{ji} = x_{ij} = 0$. 
  For all $q \in [Q]$ and all terminal sets $T^k \subseteq F^q$, we set~$z_{k\rho(q)} = 1$ and~$z_{k\ell} = 0$ for all $\ell \not= \rho(q)$.
  Observe that this assignment is well defined: As $F$ is a feasible Steiner Forest we have either $T^k \subseteq F^q$ or $T^k \cap F^q = \emptyset$ for all~$k \in [K]$ and all~$q \in [Q]$.  
  This solution~$(x^\ast, y^\ast, z^\ast)$ has an objective value of $C$.
 
  By construction, it cannot happen that~$(i,j) \in A$ \emph{and}~$(j,i) \in A$ and thus constraint~\eqref{ip:ext-dir-cut-4} is satisfied by our assignment.
  Likewise, our choice is such that~$\sum_{\ell=1}^k z^\ast_{\ell k} = 1$ for all~$k \in [K]$.
  Thus, our solution~$(x^\ast, y^\ast, z^\ast)$ satisfies constraint~\eqref{ip:ext-dir-cut-3}. 
   
  It remains to show that~$(x^\ast, y^\ast, z^\ast)$ satisfies~\eqref{ip:ext-dir-cut-1}. % and~\eqref{ip:ext-dir-cut-2}.
  Consider an arbitrary cut-set~$S \subseteq V$ and some~$k \in [K]$ together with a terminal~$t \in T^k\setminus\{r^k\}$.
  Assume that~$t \not\in S$  and $T^k\subseteq F^q$, and observe that the right-hand side of~\eqref{ip:ext-dir-cut-1} is strictly positive if and only if~$S$ contains~$r^{\rho(q)}$ because~$z_{\ell k} = 0$ for all~$\ell \not= r^{\rho(q)}$.
  In this case, however, the directed forest~$A$ contains a directed~$r^{\rho(q)}$-$t$-path by construction and thus~$y^\ast(\delta^+(S)) \geq 1$.
  Since in particular~$\sum_{\ell=1}^k z^\ast_{\ell k} \leq 1$, this means that~\eqref{ip:ext-dir-cut-1} is satisfied.
 \end{proof}
 
On the other hand, integer feasible solutions to~\eqref{ip:ext-dir-cut} imply feasible Steiner Forests.
\begin{lemma}
Let $(x^\ast, y^\ast, z^\ast)$ be a feasible integer solution to~\eqref{ip:ext-dir-cut}. 
Then, the induced graph $F:=(V, \cset{ \{i,j\} \in E}{x^\ast_{ij} = 1}$ is a feasible Steiner Forest.
\end{lemma}
\begin{proof}
Consider some terminal set $T^k$ and some terminal $t \in T^k$. 
We show that $r^k$ and $t$ lie in the same connected component of $F$.

By constraint~\eqref{ip:ext-dir-cut-3} there is exactly one $k^\ast \leq k$ with $z^\ast_{k^\ast k} = 1$.
Consider the connected component $F'$ of $F$ (induced by $x^\ast$) that contains $r^{k^\ast}$ and assume that $t \not\in F'$.
Then, in particular for $F'$, constraint~\eqref{ip:ext-dir-cut-1} yields that
\begin{align*}
 y^\ast(\delta^+(F')) \geq \slrsum{\substack{\ell \leq k:\\ r_{\ell} \in F'}} z^\ast_{\ell k} \geq z^\ast_{k^\ast k} = 1
\end{align*}

Thus, at least one edge leaving $F'$ must be contained in $F$ which contradicts the maximality of $F'$. 
We conclude that $t\in F'$.
\end{proof}

\begin{lemma}
\label{lemma:edc:vs:dc}
$\Proj_x(\LP^{dc}) \supsetneq \Proj_x(\LP^{edc})$. 
\end{lemma}
\begin{proof}[Proof (Sketch)] 
Let $(\tilde x, \tilde y, \tilde z)\in\LP^{edc}$. 
Set $\hat x := \tilde x$. 
Now, consider and fix a terminal set $k\in[K]$. 
Then, for each terminal $t\in T^k$, and each root $r^\ell$ with $\ell \leq k$ 
construct a flow $\tilde f^{\ell t}$ from $r^\ell$ to $t$ of value $z_{\ell k}$. 
Notice that if $k > 1$ we also have a flow from $r^\ell$ to $r^k$. 
Similar to the arguments and the flow construction used in the proof of Lemma \eqref{lemma:df:vs:edf}
we also consider the reversed flow $\tilde f^{k r^\ell}$ ($k > \ell$) and 
combine the flows to $\hat f^{k t} := \tilde f^{k t} + \sum_{\ell < k} (\tilde f^{k r^\ell} + \tilde f^{\ell t})$. 

It is possible to assume that $\hat f$ satisfies the following properties: 
(i) $\hat f_{ij}^{k t} \leq \tilde y_{ij}$ and $\hat f_{ji}^{k t} \leq \tilde y_{ji}, \tforall\ \{i,j\}\in E$, due to the directed cuts \eqref{ip:ext-dir-cut-1}, 
(ii) $\hat f^{k t}$ is asymmetric (as discussed in Lemma \eqref{lemma:df:vs:edf}),  
(iii), $\hat f^{k t}$ satisfies the flow conservation, 
and (iv) the flow value of $\hat f^{k t}$ is 1. 
Using this flow we set $\hat y_{ij}^k := \max_{t\in T^k}\{\hat f^{k t}_{ij}\}$. 
Due to properties (i)+(ii) it holds $\hat y_{ij}^k + y_{ji}^k \leq x_{ij}, \tforall\ \{i,j\}\in E$. 
Moreover, due to (iii)+(iv) $\hat y$ satisfies the directed cuts \eqref{ip:ext-dir-cut-1}.
Hence, $(\hat x, \hat y)$ is a feasible solution to $\LP^{dc}$ with the same solution value. 

An instance showing the strict inequality is given by Figure \ref{fig:lp-comparison}. 
\end{proof}

\subsection{A Strengthened Cut-Based Directed Formulation}

It will turn out that~\eqref{ip:ext-dir-cut} is weaker than the flow model~\eqref{ip:ext-dir-flow}.
We can retain the strength of~\eqref{ip:ext-dir-flow}, however, if we allow more variables. 
The following cut-based formulation is equivalent to~\eqref{ip:ext-dir-flow}. 
Its constraints \eqref{ip:str-ext-dir-cut-z-sum-1} and \eqref{ip:str-ext-dir-cut-z} for the correct assignment of $z$-variables are the same as in the flow-based model. 
\begin{align}
  \min\Bigsetm{c^\transp x}{(x,y,z) \in \LP^{sedc}\ \tandinteger} \tag{IPsedc} \label{ip:str-ext-dir-cut}
\end{align}
where
\begin{subequations}
\begin{align}
  \LP^{sedc} := \Bigl\{ (x,y,z) \mathrel{}\Bigm\lvert 
         y^k(\delta^+(S))            &\geq z_{k\ell} &&
                                         \begin{aligned}
                                           &\tforall\ k \in [K]\\
                                           &\tandall\ \ell \geq k\\
%                                           &\tandall\ S \in \mathfrak{S}^k_{\ell}
					&\tandall\ S\subseteq V\colon r^k \in S, S \cap T^{\ell} \not=T^\ell
                                         \end{aligned}\label{ip:str-ext-dir-cut-cuts}\\
         \sum_{\ell=1}^k z_{\ell k} &= 1 && \tforall\ k \in [K]
                                        \label{ip:str-ext-dir-cut-z-sum-1}\\
        z_{kk} &\geq z_{k\ell} &&\begin{aligned}
                                 &\tforall\ k \in [K]\setminus\{1,K\}\\
                                 &\tandall\ \ell \geq k+1
                               \end{aligned}\label{ip:str-ext-dir-cut-z}\\
         \sum_{k \in [K]} (y^k_{ij} + y^k_{ji})           &\leq x_{ij} && \tforall\ \{i,j\} \in E\label{ip:str-ext-dir-cut-capacity}\\
         y^k_{ij}, y^k_{ji}            &\in [0,1] && \begin{aligned}
                                              &\tforall\ \{i,j\} \in E\\
                                              &\tandall\ k \in [K]
                                              \end{aligned}\\
         x_{ij}    &\in [0,1] && \tforall\ \{i,j\} \in E\\
         z_{k\ell}                 &\in [0,1] && \begin{aligned} &\tforall\ k\in [K]\\ &\tandall\ \ell \geq k\end{aligned}\,\Bigr\}.
\end{align}
\end{subequations}

Constraints \eqref{ip:str-ext-dir-cut-cuts} are the directed cuts which depend on the set $k$, a second terminal set $\ell$ with $\ell\geq k$, and the related $z_{k\ell}$ variable.  
If a root node $r^k$ is an assigned parent node for terminal set $T^\ell$, i.e., $z_{k\ell} > 0$, then all directed cuts separating $r^k$ from any terminal in $T^\ell$ need to have a value of at least $z_{k\ell}$. 

Constraint \eqref{ip:str-ext-dir-cut-capacity} is a simple capacity constraint which implies that any used arc in any arborescence is payed for in the objective function.

\begin{lemma}
\label{lemma:correctness:sedc}
\eqref{ip:str-ext-dir-cut} models the Steiner Forest problem correctly. 
\end{lemma}
\begin{proof}
The proof is based on the same arguments as in the proof of Lemma \ref{lemma:correctness:ext-dir-flow}.  
Let $\tilde E \subseteq E$ be an optimal solution to the SFP. 
Variables $z$ can be assigned as before and then, $\tilde E$ can again be oriented such that each connected component is an arborescence rooted at its parent node  
giving values to variables $y^1, \ldots, y^K$ and $x$. 
Since the arborescences are disjoint it follows that constraints \eqref{ip:str-ext-dir-cut-capacity} are satisfied. 
Hence, we obtain a feasible solution to \eqref{ip:str-ext-dir-cut} with the same objective value. 

On the other hand, an optimum solution $(\tilde x, \tilde y, \tilde z)$ to \eqref{ip:str-ext-dir-cut} implies a valid hierarchy of the terminal sets. 
Moreover, constraints \eqref{ip:str-ext-dir-cut-cuts} ensure that each terminal set is connected to its parent node. 
Hence, $\tilde E := \{e\in E \mid \tilde x_e = 1\}$ is a feasible solution to the SFP with the same cost. 
\end{proof}

\begin{lemma}
\label{lemma:edf:vs:sedc}
$\Proj_x(\LP^{edf}) = \Proj_x(\LP^{sedc})$.
\end{lemma}
\begin{proof}
The constraints concerning the $z$ variables are identical in both models. 
When considering one particular terminal set $k\in[K]$ constraints \eqref{ip:ext-dir-flow-flow-conservation} model a flow of value $z_{k\ell}$ 
from $r^k$ to each terminal $t\in T^\ell$, for each $\ell \in\{k,\ldots, K\}$ (except $r^k$ itself).
On the other hand, the directed cuts \eqref{ip:str-ext-dir-cut-cuts} ensure that each directed cut separating $r^k$ and $t$ has a value of at least $z_{k\ell}$. 
This is obviously equivalent. 
Moreover, 
constraints \eqref{ip:ext-dir-flow-capacity} and \eqref{ip:ext-dir-flow-edges-capacity} on the one hand and constraint \eqref{ip:str-ext-dir-cut-capacity} on the other hand are equivalent, too.  
\end{proof}

\begin{lemma}
\label{lemma:edc:vs:sedc}
$\Proj_x(\LP^{edc}) \supsetneq \Proj_x(\LP^{sedc})$. 
\end{lemma}
\begin{proof}
Let $(\tilde x,\tilde y,\tilde z) \in \LP^{sedc}$. 
We argue that ($\hat x := \tilde x, \hat y := \sum_{k\in[K]} \tilde y^k, \hat z := \tilde z) \in \LP^{edc}$. 
Since $x$ and $z$ variables are unchanged constraints 
\eqref{ip:ext-dir-cut-3} and \eqref{ip:ext-dir-cut-z} and the variable bounds are satisfied. 
Moreover, it clearly holds \eqref{ip:ext-dir-cut-4} $\tforall\ \{i,j\}\in E$:  $\hat y_{ij} + \hat  y_{ji} = \sum_{k\in[K]} \tilde y_{ij}^k  + \sum_{k\in[K]} \tilde y_{ji}^k \leq \tilde x_{ij} = \hat x_{ij}$ due to 
\eqref{ip:str-ext-dir-cut-capacity}. 

Finally, consider a directed cut $S \subseteq V\colon S\cap T^k \not=\emptyset$ for some set $k\in[K]$. 
Notice that a cut $S$ is relevant to the sum in the right-hand side of constraint \eqref{ip:ext-dir-cut-1} 
if and only if it is a valid cut for constraint \eqref{ip:str-ext-dir-cut-cuts} ($\ell$ and $k$ are interchanged in both constraints). 
Hence, it holds: 
\[
\hat y(\delta^+(S)) = \sum_{\ell=1}^K \tilde y^\ell(\delta^+(S)) \geq \sum_{\ell=1}^k \tilde y^\ell(\delta^+(S)) 
\geq \srsum{\substack{\ell\leq k:\\ r^\ell \in S}} \tilde z_{\ell k} = \srsum{\substack{\ell\leq k:\\ r^\ell \in S}} \hat z_{\ell k} 
\]

An example for strict inequality, i.e., $x \in \Proj_x(\LP^{edc})$ and $x\not\in\Proj_x(\LP^{sedc})$, is given by Figure \ref{fig:lp-comparison}. 
\end{proof}
We summarize the results of the discussion in Figure~\ref{fig:lp-relaxations}.

\begin{figure}
\centering
\begin{tikzpicture}[
  ipnode/.style={draw},
  rel/.style={->, >=stealth},
  rellabel/.style={fill=white, rectangle, font=\scriptsize}
  ]
  \draw ( 0, 0) node[ipnode]   (ipundir)   {$\Proj_x (\LP^{uf}) = \LP^{uc}$};
  \draw (-6, 0) node[ipnode]   (ipmr)      {$\Proj_x(\LP^{mr})$};
  \draw ( 6, 0) node[ipnode]   (ipklsvz)   {$\Proj_x(\LP^{klsvz})$};
  \draw ( 0,-2) node[ipnode]   (ipdir)     {$\Proj_x(\LP^{df}) = \LP^{dc}$};
  \draw ( 4,-4) node[ipnode]   (ipext)     {$\LP^{edc}$};
  \draw ( 0,-6) node[ipnode, label=right:{\scriptsize Lemma\ \ref{lemma:edf:vs:sedc}}]   (ipsext)    {$\Proj_x(\LP^{edf}) = \Proj_x(\LP^{sedc})$}; 
  
  \draw[rel] (ipundir) -- node[rellabel] {Lemma~\ref{lemma:uc:vs:mr}}    (ipmr);
  \draw[rel] (ipundir) -- node[rellabel] {Lemma~\ref{lemma:uc:vs:klsvz}} (ipklsvz);
  \draw[rel] (ipundir) -- node[rellabel] {Lemma~\ref{lemma:uf:vs:df}}    (ipdir);

  \draw[rel] (ipdir)   -- node[rellabel] {Lemma~\ref{lemma:edc:vs:dc}}   (ipext);

  \draw[rel] (ipdir)   -- node[rellabel] {Lemma~\ref{lemma:df:vs:edf}}   (ipsext);

  \draw[rel] (ipext)   -- node[rellabel] {Lemma~\ref{lemma:edc:vs:sedc}} (ipsext);
\end{tikzpicture}
\caption{\label{fig:lp-relaxations}%
  Relationship between the different linear programming relaxations. The arrows point in the direction of the stronger relaxation.}
\end{figure}
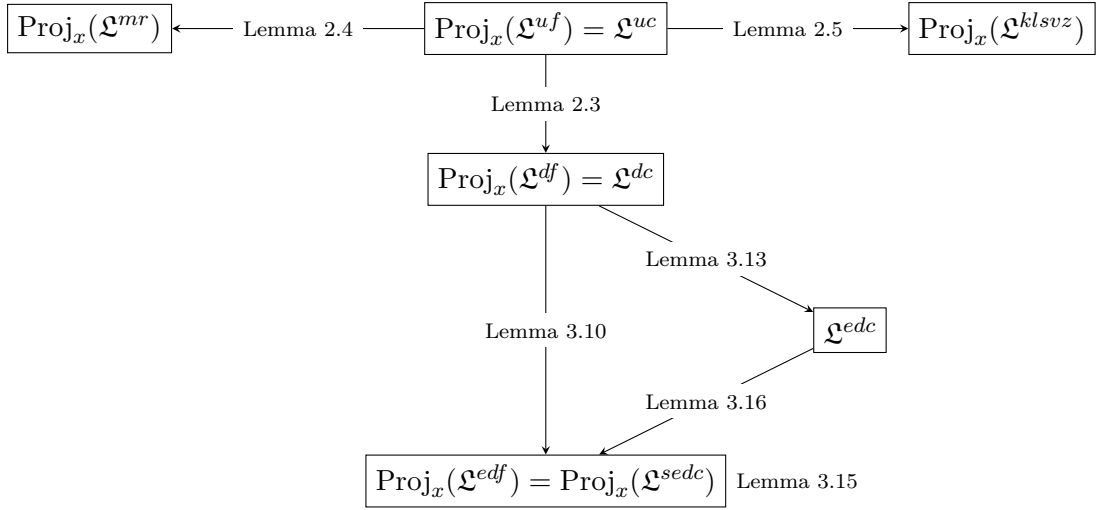

\section{Experimental Evaluation}\label{sec:experiments}

Figure~\ref{fig:lp-relaxations} in the previous section shows that the linear programming bound obtained from $\LP^{sedc}$ is never worse than the classical linear programming bounds for the Steiner Forest Problem. 
It does not tell us, however, by how much the new bound is better. 
While we cannot provide a theoretical guarantee on the quality of the bound, we evaluate its practical usefulness -- and the practical usefulness of the other linear programming formulations -- through computational experiments.   

\subsection{Setting}
We performed the experiments on a Debian 8 system with a \texttt{Intel(R) Xeon(R) E5-2690 v2} CPU running at 3 GHz.
Our implementation is in \texttt{C++} using \texttt{CPLEX 12.6.2} and was compiled with \texttt{g++-4.9.2} using the \texttt{-O2} flag.

\subsection{Details of the Implementation}

The formulations were solved as relaxations, i.e., all variables were set to be continous. 
We disabled the presolving and symmetry breaking capabilities of \texttt{CPLEX}.
Wherever it was useful and a polynomial time separation algorithm was available, we generated the constraints of the formulations dynamically:
\begin{itemize}
  \item The separation for the cut-set constraints~\eqref{ip:undir-cut-1} and~\eqref{ip:dir-cut-1} in~\eqref{ip:undir-cut} and in~\eqref{ip:dir-cut}, respectively, is standard; we use a minimum-$s$-$t$-cut procedure.
  \item In formulation~\eqref{ip:klsvz}, we separate the cut constraints~\eqref{ip:klsvz-1} and~\eqref{ip:klsvz-2} through repeated calls to a minimum-$s$-$t$-cut procedure. 
  \item We separate the subtour elimination constraints~\eqref{ip:ext-tree-2} in formulation~\eqref{ip:ext-tree} constraints with the standard construction. 
  We separate the three-cycle-inequalities~\eqref{ip:ext-tree-9}--\eqref{ip:ext-tree-11} through complete enumeration in time $O(|V|^3)$. 
  \item  We separate the cut-set constraints~\eqref{ip:ext-dir-cut-1} and~\eqref{ip:str-ext-dir-cut-cuts} in~\eqref{ip:ext-dir-cut} and in~\eqref{ip:str-ext-dir-cut}, respectively, with a minimum-$s$-$t$-cut procedure.
\end{itemize} 

\subsection{Benchmark Instances}

We generated 460 random network topologies using a method by Johnson, Minkoffs and Philipps~\cite{JMP2000}: 
First, distribute $n$ nodes uniformly at random in a unit square. 
Then, connect any two nodes $i$ and $j$ with an edge $\{i,j\}$ if their Euclidean distance is less than $\alpha / \sqrt{n}$, where $\alpha$ is a parameter for the random generator.
The cost of the edge $\{i,j\}$ is proportional to the Euclidean distance.
Finally, connect all nodes with a minimum Euclidean spanning tree to ensure that the instance is connected. 

To determine $k$ random terminal sets, we first select $p \cdot |V|$ nodes uniformly at random (the number $k\in [n/2]$ of terminal sets and the terminal percentage $p \in [0,1]$ are again parameters).
We then bring the selected nodes into a random order and draw $k-1$ distinct split points from $\{2,\dots,k-1\}$, thus splitting the random node order into $k$ distinct terminal sets.

For each $n \in \{25, 50, 150, 200\}$, we choose a small, a medium, and a large number of terminal sets $k$ (see Table~\ref{tab:nsolved} for details).
The percentage $t$ of terminal nodes is picked from $\{0.25, 0.5,  0.75, 1.0\}$ unless a combination of $n, k$ and $t$ results in a terminal set size of less than two. 
For each choice of $n$, $k$, and $t$, we generate five instances with $\alpha=1.6$ and five instances with $\alpha=2.0$.

\subsection{Solveability of the Linear Programming Relaxations}
In a first step, we evaluate the solveability of the linear programming models from the previous sections.
Figure~\ref{fig:solved-over-time} shows on the y-axis on how many instances the \emph{linear programming relaxation} of each model has been solved to optimality after a given amount of time on the x-axis.
There are 460 instances and a time limit of one hour. 
As a general trend, the largest part of the relaxations is solved within the first 400 seconds.
However, none of the models allows for the linear programming relaxation of \emph{all} instances to be solved: 
We observe that from the previously known models, only the na\"ive undirected cut-set formulation~\eqref{ip:undir-cut} allows the LP-solver to find the optimum solutions for the bulk of the relaxations. 
The new models~\eqref{ip:ext-dir-cut} and \eqref{ip:str-ext-dir-cut} solve a comparable number of relaxations to optimality and do so much faster than the na\"ive formulation~\eqref{ip:undir-cut}.
Out of the improvements of the na\"ive formulation~\eqref{ip:undir-cut}, mode~\eqref{ip:klsvz} performs best in terms of speed and number of instances solved. 
Using this model, the LP-solver finds optimum solutions to roughly 80\% of the instances. 
The directed cut-set formulation~\eqref{ip:dir-cut} and the new model~\eqref{ip:ext-tree} yield optimum solutions on around 55\% of the instances, with~\eqref{ip:dir-cut} performing slightly worse.
Finally, using the model~\eqref{ip:mr}, the LP-solver finds optimum solutions on less than 25\% of the relaxations.   

The time limit being chosen fairly generously, we identified three main causes for the solver to fail to solve a relaxation: 
First, the available memory was not sufficient to even build the (initial) model. 
Second, the initial model could be build, but the memory was not sufficient to add all the necessary constraints during separation.
Third, the separation process had not added all necessary constraints when the time limit was reached.
In case of the first failure type, we cannot extract any lower bound from the model. 
This type mainly occured on the static model~\eqref{ip:mr} where no separatation procedure could be used. 
Whether the constraints of type~\eqref{ip:mr-4}, \eqref{ip:mr-5}, and~\eqref{ip:mr-7} can be separated efficiently is an interesting open question.
In the other cases, suboptimal lower bounds are available.
Table~\ref{tab:nsolved} shows in how many cases we were able to extract \emph{any} bound from a linear programming relaxation
for each choice of the parameters $n$ and $k$, and a model $m$.
The table also shows in how many cases we obtained the optimum bound. 
In this sense, the table gives a more detailed picture on the situation in Figure~\ref{fig:solved-over-time}: 
The larger $n$ and $k$, the harder the relaxations are to solve. 

Let us first discuss model~\eqref{ip:mr}, as this model is the only one without dynamic constraint generation.
As could be expected given its size, the performance of~\eqref{ip:mr} degrades rapidly as $n$ or $k$ increase: 
It fails to reliably provide a bound on all but the smallest instances.
For $n > 25$ of if $k > 3$, we cannot expect to solve this model to optimality anymore.
Starting from $n=100$, no more bounds can be found using this model. 
All other models use a separation procedure and yield a valid bound in all cases.
While na\"ive cut-set model~\eqref{ip:undir-cut}, and the new models~\eqref{ip:ext-dir-cut} and~\eqref{ip:str-ext-dir-cut} yield optimum bounds even for the largest instances, the models~\eqref{ip:dir-cut} and \eqref{ip:ext-tree} start failing at $n=100$.
The model~\eqref{ip:klsvz} provides optimum bounds for $n \leq 150$.

\begin{landscape}
\begin{figure}
  \begin{tikzpicture}
  \input{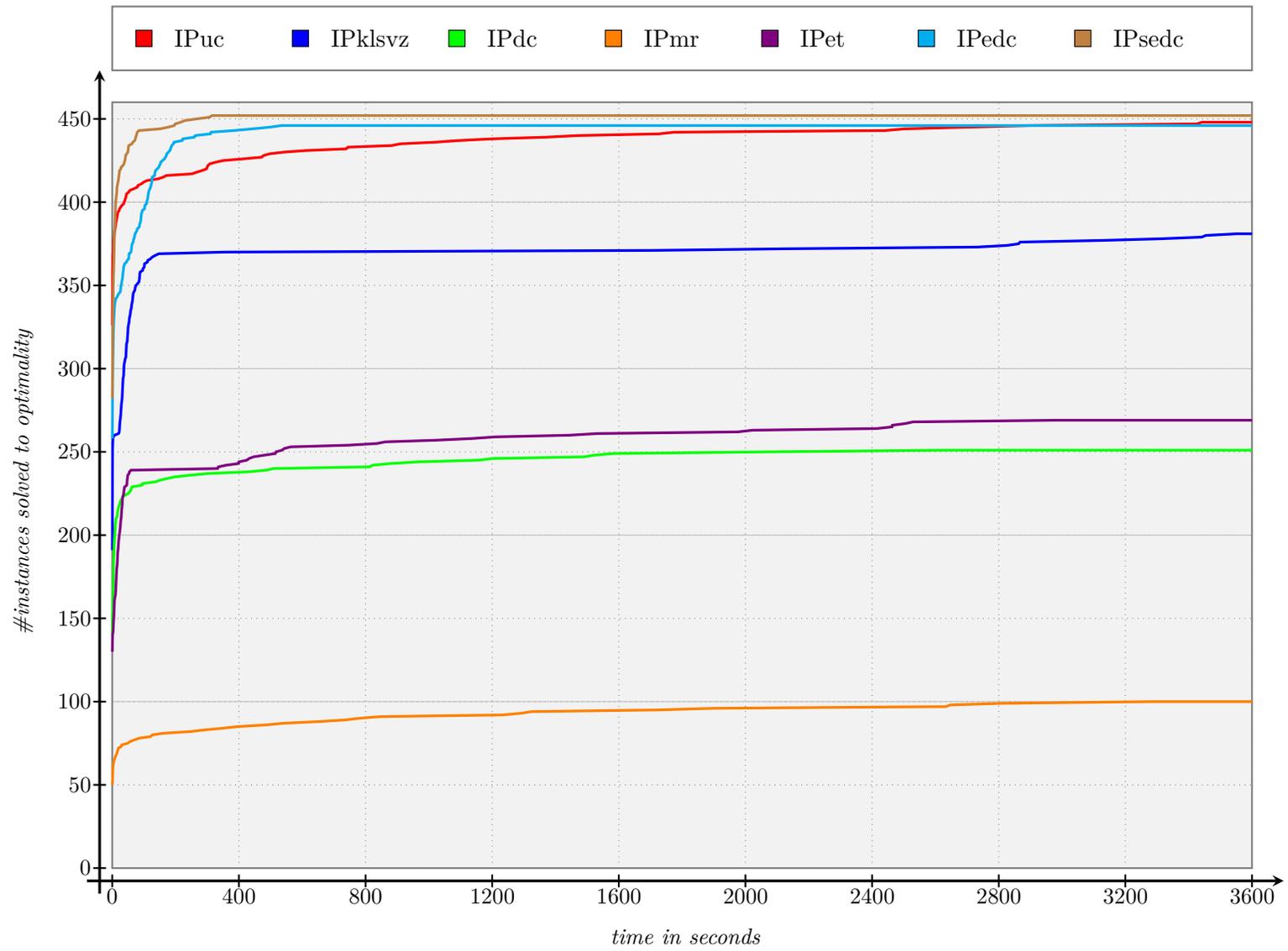}
  \end{tikzpicture}
  \caption{Number of instances solved to optimality after $x$ seconds.\label{fig:solved-over-time}}
\end{figure}
\end{landscape}

\begin{landscape}
\begin{table}[p]
  % -- number of solved instances (out of 40)  -------------
\begin{tabular}{rrrrrrrrrrrrrrrr}
\toprule
    &    & \multicolumn{2}{c}{\ref{ip:undir-cut}} & \multicolumn{2}{c}{\ref{ip:dir-cut}} & \multicolumn{2}{c}{\ref{ip:klsvz}} & \multicolumn{2}{c}{\ref{ip:mr}} & \multicolumn{2}{c}{\ref{ip:ext-tree}} & \multicolumn{2}{c}{\ref{ip:ext-dir-cut}} & \multicolumn{2}{c}{\ref{ip:str-ext-dir-cut}} \\
$n$ & $k$ &       feas. &       opt. &      feas. &       opt. &      feas. &       opt. &      feas. &       opt. &      feas. &       opt. &      feas. &       opt. &      feas. &       opt. \\
\cmidrule(lr){1-2} \cmidrule(lr){3-4} \cmidrule(lr){5-6} \cmidrule(lr){7-8} \cmidrule(lr){9-10} \cmidrule(lr){11-12} \cmidrule(lr){13-14} \cmidrule(lr){15-16}  25 &   2 &       100\% &     100\% &      100\% &     100\% &      100\% &     100\% &      100\% &     100\% &      100\% &     100\% &      100\% &     100\% &      100\% &     100\% \\
 25 &   3 &       100\% &     100\% &      100\% &     100\% &      100\% &     100\% &      100\% &      75\% &      100\% &     100\% &      100\% &     100\% &      100\% &     100\% \\
 25 &   4 &       100\% &     100\% &      100\% &     100\% &      100\% &     100\% &       67\% &      37\% &      100\% &     100\% &      100\% &     100\% &      100\% &     100\% \\
\\
 50 &   3 &       100\% &     100\% &      100\% &     100\% &      100\% &     100\% &       60\% &      32\% &      100\% &     100\% &      100\% &     100\% &      100\% &     100\% \\
 50 &   4 &       100\% &     100\% &      100\% &     100\% &      100\% &     100\% &       35\% &      12\% &      100\% &      98\% &      100\% &     100\% &      100\% &     100\% \\
 50 &   5 &       100\% &     100\% &      100\% &     100\% &      100\% &     100\% &       20\% &       2\% &      100\% &     100\% &      100\% &     100\% &      100\% &     100\% \\
\\
100 &   5 &       100\% &     100\% &      100\% &      52\% &      100\% &     100\% &        0\% &       0\% &      100\% &      70\% &      100\% &     100\% &      100\% &     100\% \\
100 &  10 &       100\% &     100\% &      100\% &       0\% &      100\% &     100\% &        0\% &       0\% &      100\% &      28\% &      100\% &     100\% &      100\% &      98\% \\
100 &  15 &       100\% &     100\% &      100\% &       0\% &      100\% &     100\% &        0\% &       0\% &      100\% &       3\% &      100\% &     100\% &      100\% &     100\% \\
\\
200 &  10 &       100\% &      95\% &      100\% &       0\% &      100\% &      32\% &        0\% &       0\% &      100\% &       0\% &      100\% &      90\% &      100\% &      98\% \\
200 &  15 &       100\% &      92\% &      100\% &       0\% &      100\% &      38\% &        0\% &       0\% &      100\% &       0\% &      100\% &      88\% &      100\% &      95\% \\
200 &  20 &       100\% &      82\% &      100\% &       0\% &      100\% &      40\% &        0\% &       0\% &      100\% &       0\% &      100\% &      88\% &      100\% &      90\% \\
\bottomrule
\end{tabular}

\caption{Percentage of instances where a valid bound (an optimum bound) was found within one hour.\label{tab:nsolved}}
\end{table}
\end{landscape}

\subsection{Quality of the Linear Programming Bounds}
We now argue that the new models not only solve fast, but also provide good bounds. 
The easiest way to evaluate the quality of a bound would be to compare against the integer optimum. 
However, this optimum is not known in all cases. 
As an alternative, we look at the improvement over the bound obtained from the na\"ive model~\eqref{ip:undir-cut}: 
This model is fast and at worst yields a bound that is a 2-factor approximation for the integer optimum. 

It turns out that in our experiments, the bounds from~\eqref{ip:klsvz} are the same as the ones from the na\"ive model~\eqref{ip:undir-cut}.
We therefore do not show model~\eqref{ip:klsvz} in the following figures. 
Likewise, we removed model~\eqref{ip:mr} from the comparison as it did not provide a significant number of bounds. 
It remains to compare the known improvement~\eqref{ip:dir-cut} with the new models.

Figure~\ref{fig:bound-improvement} shows the factor by which the bounds obtained from \eqref{ip:dir-cut} with the new models~\eqref{ip:ext-dir-cut}, \eqref{ip:str-ext-dir-cut}, and~\eqref{ip:ext-tree} differ from the bound provided by~\eqref{ip:undir-cut}.
Recall that~\eqref{ip:undir-cut} is a guaranteed 2-approximation such that the integer optimum is at worst at the 200\% line in Figure~\ref{fig:bound-improvement}.
The figure is a standard box plot diagram in which each box corresponds to one model and aggregates all instances.
The lower and the upper whisker at each box show the maximum and the minimum factor throughout all the instances, respectively. 
Each box has a lower and an upper end; these depict the 25\%-percentile and the 75\%-percentile of the improvement factors, respectively. 
The thick line in each box shows the median of the improvement factors. 
Consider for instance the case where $n=25$ and focus on the model~\eqref{ip:dir-cut}: 
The lower whisker shows that the model provided a bound that was at least as good as the bound from~\eqref{ip:undir-cut} on all instances.
As shown by the upper whisker, the best bound found throughout all of the instances was a factor of 1.6 better than the bound from~\eqref{ip:undir-cut}.
The box itself shows that on 25\% of the instances, the factor was at least 1.4 (upper end of the box), on 50\% of the instances, the factor was at least 1.25 (median line), and on 75\% of the instances, the factor was at least 1.15 (lower end of the box).

Model~\eqref{ip:dir-cut} clearly improves on the na\"ive formulation. 
While the improvement tapers off towards the larger instances, we see a bound that is better by a factor of roughly 1.2 even for $n=200$.
We remark that~\eqref{ip:dir-cut} is no longer able to solve all linear programs to optimality for larger $n$ and thus expect the theoretical best bound provided by the model to be slightly better than depicted here. 
The median improvement of model~\eqref{ip:ext-tree} is comparable to the the one of model~\eqref{ip:dir-cut}. 
Its variance is much higher, though: While some bounds of its improve on~\eqref{ip:dir-cut}, others are far worse than the na\"ive bounds from~\eqref{ip:undir-cut}.
This observation holds true even for $n=25$ and $n=50$ where~\eqref{ip:ext-tree} solves to optimality on all instance.
We conclude that even optimum bounds from~\eqref{ip:ext-tree} can be very weak.
A more detailed analysis (not shown here) yields that the bounds from~\eqref{ip:ext-tree} deteriorate with a growing number of Steiner nodes: 
In fact, the bound from~\eqref{ip:ext-tree} is never worse than the na\"ive bound from~\eqref{ip:undir-cut} on the instances where all nodes are terminals.    

The new models~\eqref{ip:ext-dir-cut} and~\eqref{ip:str-ext-dir-cut} perform similarly. 
For $n=25$, their bound is never worse then the na\"ive bound and their median lies well above the median of the previous best model~\eqref{ip:dir-cut}.
The trend continues for larger $n$ where the bounds from these two models is strictly better than the na\"ive bound in all cases. 
We even see an increasing improvement that comes close to a factor of $2$ in the best case.
Model~\eqref{ip:str-ext-dir-cut} seems to perform slightly better than~\eqref{ip:ext-dir-cut}; we recall however, that~\eqref{ip:str-ext-dir-cut} solves to optimality more often for $n=200$ where the difference is most pronounced.     

\begin{figure}[t]
\begin{tikzpicture}
 \input{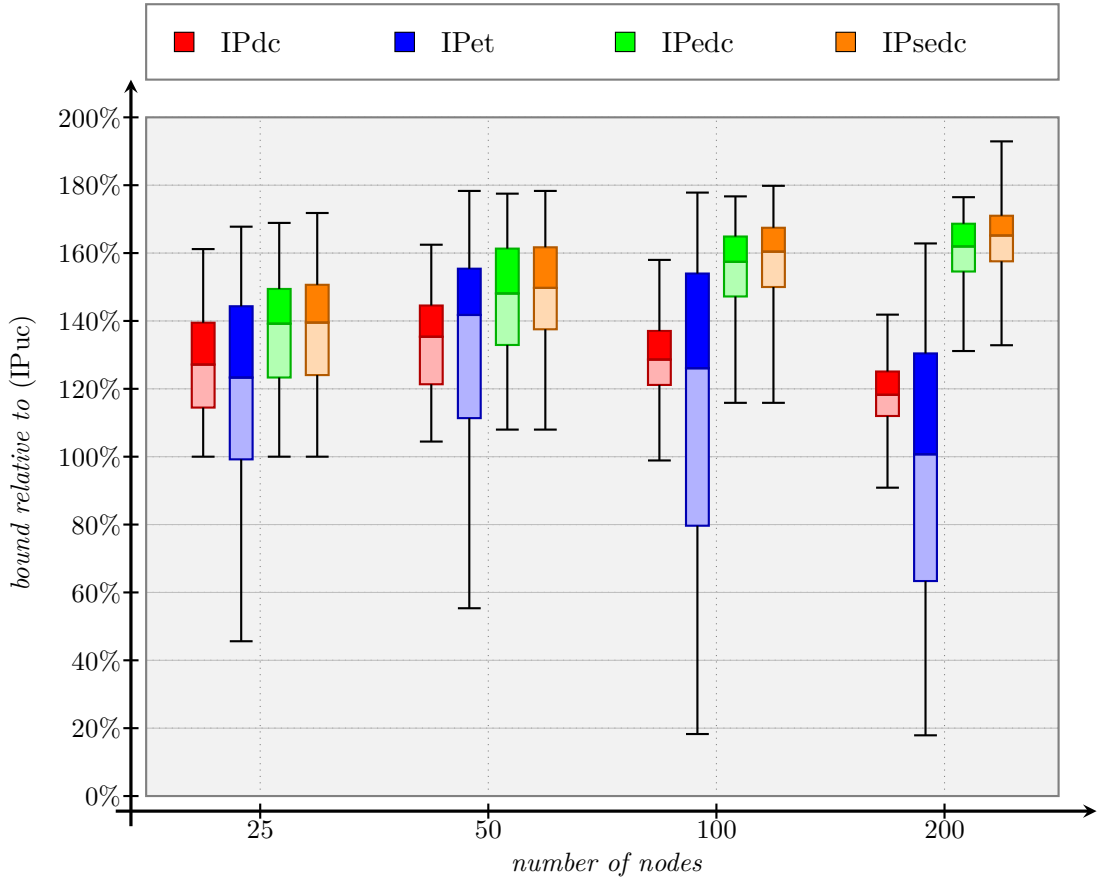}
\end{tikzpicture}
\caption{\label{fig:bound-improvement}
Quality of the linear programming bound of the best previous model and the new models. 
The values show the improvement over the linear programming bound over the naive cut-set formulation~\eqref{ip:undir-cut} at 100\%. 
The figure shows the (possibly suboptimal) bound obtained after 1 hour.}
\end{figure}

\section{Conclusions}
While our tree-based formulation~\eqref{ip:ext-tree} only seems to work well for instances without Steiner nodes, the new cut-based formulations~\eqref{ip:ext-dir-cut} and~\eqref{ip:str-ext-dir-cut} are tractable and provide lower bounds that come close to the integer optimum in many cases.
Switching between~\eqref{ip:ext-dir-cut} and~\eqref{ip:str-ext-dir-cut} allows us to trade a smaller number of variables for a slightly worse lower bound.
The obvious next step would be to design Branch-and-Bound algorithms based on these formulations. 
Another interesting question is whether the new formulations imply a primal-dual approximation algorithm and to find theoretical quality guarantees for them. 
Finally, it would be interesting to shed more lightonto the connection between~\eqref{ip:mr} and~\eqref{ip:ext-dir-cut} or~\eqref{ip:str-ext-dir-cut}. 
What happens if only a subset of the constraints~\eqref{ip:mr-7} is used, and which would be a good candidate subset? 
Can the constraints~\eqref{ip:mr-7} be separated efficiently?

\paragraph*{Acknowledgements}
This work was supported by a fellowship within the Postdoc-Program of the German Academic Exchange Service (DAAD).

\bibliographystyle{amsalpha}
\bibliography{steiner-forest}

\providecommand{\bysame}{\leavevmode\hbox to3em{\hrulefill}\thinspace}
\providecommand{\MR}{\relax\ifhmode\unskip\space\fi MR }
\providecommand{\MRhref}[2]{%
  \href{http://www.ams.org/mathscinet-getitem?mr=#1}{#2}
}
\providecommand{\href}[2]{#2}
\begin{thebibliography}{BMW89}

\bibitem[BMW89]{BMW1989}
{A}. Balakrishnan, {T}.~{L}. {Magnanti}, and {R}.~{T}. {Wong}, \emph{{A}
  {Dual}-{Ascent} {Procedure} for {Large}-{Scale} {Uncapacitated} {Network}
  {Design}}, Oper. {Res}. \textbf{37} (1989), no.~5, 716--740.

\bibitem[CR94]{CR1994}
S.~Chopra and M.~R Rao, \emph{{The Steiner tree problem I: Formulations,
  compositions and extension of facets}}, Mathematical Programming \textbf{64}
  (1994), no.~1, 209--229.

\bibitem[Edm03]{Edmonds2003}
{Jack} Edmonds, \emph{Submodular {Functions}, {Matroids}, and {Certain}
  {Polyhedra}}, Combinatorial {Optimization} {\textemdash} {Eureka}, {You}
  {Shrink}! ({Michael} {J}{\"u}nger, {Gerhard} {Reinelt}, and {Giovanni}
  {Rinaldi}, eds.), Lecture {Notes} in {Computer} {Science}, no. 2570, Springer
  {Berlin} {Heidelberg}, 2003, pp.~11--26.

\bibitem[Gal57]{Gale1957}
D.~Gale, \emph{A theorem on flows in networks.}, Pacific Journal of Mathematics
  \textbf{7} (1957), no.~2, 1073--1082.

\bibitem[GM93]{GM1993}
{Michel}~{X}. Goemans and {Young}-{Soo} {Myung}, \emph{{A} catalog of steiner
  tree formulations}, Networks \textbf{23} (1993), no.~1, 19--28.

\bibitem[Goe94]{Goemans1994}
{Michel}~{X}. Goemans, \emph{The {Steiner} tree polytope and related
  polyhedra}, Mathematical {Programming} \textbf{63} (1994), no.~1-3, 157--182.

\bibitem[JMP00]{JMP2000}
D.~S. Johnson, M.~Minkoff, and S.~Phillips, \emph{{The Prize Collecting Steiner
  Tree Problem: Theory and Practice}}, {Proceedings of the SODA}, {SODA '00},
  SIAM, 2000, pp.~760--769.

\bibitem[KLSv08]{KLSvZ2008a}
J.~K{\"o}nemann, S.~Leonardi, G.~Sch{\"a}fer, and S.~{van Zwam}, \emph{A
  {{Group}}-{{Strategyproof Cost Sharing Mechanism}} for the {{Steiner Forest
  Game}}}, SIAM Journal on Computing \textbf{37} (2008), no.~5, 1319--1341.

\bibitem[Luc93]{Lucena1993}
{Abilio} Lucena, \emph{Tight bounds for the {Steiner} problem in graphs},
  Technical report, RC for Process Systems Engineering, Imperial College,
  London, 1993.

\bibitem[Mar86]{Martin1986}
{R}.~{K}. Martin, \emph{{A} sharp polynomial size linear programming
  formulation of the minimum spanning tree problem.}, Working {Paper},
  University of {Chicago}, 1986.

\bibitem[MPL94]{MPL1994}
{F}. Margot, {A}. {Prodon}, and {Th}~{M}. {Liebling}, \emph{Tree polytope on
  2-trees}, Mathematical {Programming} \textbf{63} (1994), no.~1-3, 183--191.

\bibitem[MR05]{MR2005}
T.~L Magnanti and S.~Raghavan, \emph{Strong formulations for network design
  problems with connectivity requirements}, Networks \textbf{45} (2005),
  61--79.

\bibitem[Rag95]{Raghavan1995}
S.~Raghavan, \emph{Formulations and algorithms for network design problems with
  connectivity requirements}, Phd thesis, Massachusetts Institute of
  Technology, 1995.

\end{thebibliography}
\end{document}